\documentclass[final,a4paper,UKenglish,cleveref, numberwithinsect, autoref, thm-restate]{lipics-v2021}
\pdfoutput=1

\usepackage{xspace}
\usepackage{ifthen}
\usepackage{booktabs}
\usepackage{makecell}
\usepackage{microtype}

\newboolean{withappendix}
\setboolean{withappendix}{true}

\hypersetup{hidelinks,final}

\newcommand{\resetCurThmBraces}{%
\gdef\curThmBraceOpen{(}%
\gdef\curThmBraceClose{)}}
\resetCurThmBraces
\newcommand{\removeThmBraces}{%
\gdef\curThmBraceOpen{}%
\gdef\curThmBraceClose{}}
\resetCurThmBraces

\newenvironment{notheorembrackets}{\removeThmBraces}{\resetCurThmBraces}

\usepackage{etoolbox}
\patchcmd{\thmhead}{(#3)}{\curThmBraceOpen #3\curThmBraceClose }{}{}

\newcommand{\hackynewpage}{%
    \pagebreak%
}

\usepackage{tikz}
\usetikzlibrary{cd,calc}

\tikzset{shiftarr/.style={
        rounded corners,%
        to path={--([#1]\tikztostart.center)
                     -- ([#1]\tikztotarget.center) \tikztonodes
                     -- (\tikztotarget)},
}}

\newsavebox{\mypullbackcorner}%
\sbox{\mypullbackcorner}{%
\begin{tikzpicture}
    \draw[-] (0,0) -- (.5em,.5em) -- (0,1em);
\end{tikzpicture}%
}
\newcommand{\pullbackangle}[2][]{\arrow[phantom,to path={
                     -- ($ (\tikztostart)!1cm!#2:([xshift=8cm]\tikztostart) $)
                        node[anchor=west,pos=0.0,rotate=#2,
                        inner xsep = 0]
                        {\begin{tikzpicture}[minimum
                        height=1mm,baseline=0,#1]
    \draw[-] (0,0) -- (.5em,.5em) -- (0,1em);
                        \end{tikzpicture}}}]{}}

\tikzset{
  coalgebra/.style={
    block line/.style={
      draw=black!50,
      line width=1.2pt,
    },
    block/.style={
      block line,
      rounded corners=3pt,
      inner sep=1pt,
      minimum height=6mm,
      minimum width=6mm,
    },
    scissors line/.style={
      draw=black!50,
      text=black!50,
      font=\footnotesize,
      line width=0.8pt,
      shorten <= -4pt,
      shorten >= -4pt,
      dotted,
    },
    state/.style={
      text depth=0pt,
      outer sep=0pt,
      inner sep=4pt,
    },
    transition/.style={
      -{latex},
      line width=0.8pt,
      draw=black,
      preaction = {draw,-,draw=white,line width=4.6pt,line cap=round},
    },
    path with edges/.style={
      every edge/.append style={transition}
    },
  },
}

\newcommand{\C}{\ensuremath{\mathcal{C}}}
\newcommand{\E}{\ensuremath{\mathcal{E}}}
\newcommand{\R}{\ensuremath{\mathbb{R}}}
\newcommand{\N}{\ensuremath{\mathbb{N}}}
\newcommand{\K}{\mathcal{K}}%
\newcommand{\D}{\ensuremath{\mathcal{D}}}
\newcommand{\CO}{\ensuremath{\mathcal{O}}}
\newcommand{\M}{\ensuremath{\mathcal{M}}}
\newcommand{\Set}{\ensuremath{\mathsf{Set}\xspace}}
\newcommand{\Pow}{\ensuremath{\mathcal{P}}}
\newcommand{\Bag}{\ensuremath{\mathcal{B}}}
\newcommand{\colim}{\ensuremath{\operatorname{colim}}}
\newcommand{\dual}{{\ensuremath{\operatorname{\mathsf{op}}}}}
\newcommand{\Coalg}{\ensuremath{\mathsf{Coalg}}}
\newcommand{\Alg}{\ensuremath{\mathsf{Alg}}}

\newcommand{\StrongEpi}{\ensuremath{\mathsf{StrongEpi}}}
\newcommand{\StrongMono}{\ensuremath{\mathsf{StrongMono}}}
\newcommand{\Mor}{\ensuremath{\mathsf{Mor}}}
\newcommand{\Iso}{\ensuremath{\mathsf{Iso}}}
\newcommand{\Epi}{\ensuremath{\mathsf{Epi}}}

\newcommand{\Mono}{\ensuremath{\mathsf{Mono}}}

\newcommand{\reach}{\ensuremath{\mathsf{reach}}}
\newcommand{\pr}{\ensuremath{\mathsf{pr}}}

\newcommand{\inj}{\ensuremath{\mathsf{inj}}}
\newcommand{\id}{\ensuremath{\mathsf{id}}}
\newcommand{\Id}{\ensuremath{\mathsf{Id}}}
\renewcommand{\Im}{\ensuremath{\mathsf{Im}}}

\newcommand{\monoto}{\ensuremath{\rightarrowtail}}
\newcommand{\hookto}{\ensuremath{\hookrightarrow}}
\newcommand{\epito}{\ensuremath{\twoheadrightarrow}}

\newcommand{\takeout}[1]{\relax}

\newcommand{\itemref}[2]{\autoref{#1}.\ref{#2}}
\newcommand{\textqt}[1]{`#1'}
\newcommand{\singlequote}[1]{`#1'}
\newcommand{\etal}{\text{et al.}}
\newcommand{\Adamek}{Ad{\'{a}}mek}

\newcommand{\textshaded}[1]{\text{#1}}

\newenvironment{proofappendix}[2][Proof of]{%
\subsubsection*{#1~\autoref{#2}}%
\addcontentsline{toc}{subsection}{#1~\autoref{#2}}%
}{}

\newenvironment{listinenv}{\setlength{\parskip}{0.2em plus 0.1em minus
    0em}\leavevmode}%
{\setlength{\parskip}{0.0em plus 0.6em minus 0.0em}}

\makeatletter 
\newcommand\mynobreakpar{\par\nobreak\@afterheading} 
\makeatother

\newcommand{\descto}[3][]{\arrow[phantom]{#2}[#1]{\text{\footnotesize{}\begin{tabular}{c}#3\end{tabular}}}}

\theoremstyle{plain}

\theoremstyle{definition}

\newtheorem{assumption}[theorem]{Assumption}
\newtheorem{instance}[theorem]{Instance}

\title{Minimality Notions via Factorization Systems}
\titlerunning{Minimality Notions via Factorization Systems} %

\author{Thorsten Wißmann}{Radboud University, Nijmegen, The Netherlands\and\url{https://thorsten-wissmann.de}}{thorsten.wissmann@ru.nl}{https://orcid.org/0000-0001-8993-6486}{%
}

\authorrunning{T.~Wißmann}

\Copyright{Thorsten Wißmann} %



\category{(Co)algebraic pearls}
\relatedversiondetails{Full Version}{https://arxiv.org/abs/2106.07233}

\funding{Supported by the NWO TOP project 612.001.852}%

\acknowledgements{%
The author thanks Stefan Milius and Jurriaan
Rot for inspiring discussions and thanks the referees for their helpful comments.
}%

\nolinenumbers %

\EventEditors{Fabio Gadducci and Alexandra Silva}
\EventNoEds{2}
\EventLongTitle{9th Conference on Algebra and Coalgebra in Computer Science (CALCO 2021)}
\EventShortTitle{CALCO 2021}
\EventAcronym{CALCO}
\EventYear{2021}
\EventDate{August 31--September 3, 2021}
\EventLocation{Salzburg, Austria}
\EventLogo{}
\SeriesVolume{211}
\ArticleNo{24}

\begin{document}
\maketitle
\begin{abstract}
  For the minimization of state-based systems (i.e.\ the reduction of the number
  of states while retaining the system's semantics), there are two obvious
  aspects: removing unnecessary states of the system and merging redundant
  states in the system. In the present article, we relate the two aspects on
  coalgebras by defining an abstract notion of minimality.

  The abstract notion minimality and minimization live in a general category
  with a factorization system. We will find criteria on the category that ensure
  uniqueness, existence, and functoriality of the minimization aspects. The
  proofs of these results instantiate to those for reachability and observability minimization
  in the standard coalgebra literature. Finally, we will see how the two aspects
  of minimization interact and under which criteria they can be sequenced in any
  order, like in automata minimization.
\end{abstract}
\section{Introduction}
Minimization is a standard task in computer science that comes in different
aspects and lead to various algorithmic challenges. The task is to reduce the
size of a given system while retaining its semantics, and in general there are
two aspects of making the system smaller: 1.~merge redundant parts of the system
that exhibit the same behaviour (\emph{observability}) and 2.~omit
unnecessary parts (\emph{reachability}). Hopcroft's automata minimization
algorithm~\cite{Hopcroft71} is an early example: in a given deterministic
automaton, 1.~states accepting the same language are identified and
2.~unreachable states are removed. Moreover, Hopcroft's algorithm runs in
quasilinear time; for an automaton with $n$ states, reachability is computed 
in $\CO(n)$ and observability in $\CO(n\log n)$.

Since the reachability is a simple depth-first search, it is straightforward to
apply it to other system types. On the other hand, it took decades until
quasilinear minimization algorithms for observability were developed for other
system types such as transition systems~\cite{PaigeT87}, labelled transition
systems~\cite{DovierPP04,Valmari09}, or Markov
chains~\cite{DerisaviHS03,ValmariF10}. Though their differences in complexity,
the aspects of observability and reachability have very much in common when
modelling state-based systems as coalgebras. Then, observability is the task to
find the greatest coalgebra quotient and reachability is the task of finding the
smallest subcoalgebra containing the initial state, or generally, a
distinguished point of interest.

In the present article, we define an abstract notion of minimality and
minimization in a category with an $(\E,\M)$-factorization system. Such a
factorization systems gives rise to a generalized notion of quotients and subobjects, and
the minimization is the task of finding the least quotient resp.~subobject.
To make this general setting applicable to coalgebras, we show that the category of
coalgebras inherits the factorization system from the base category under a mild
assumption (namely that the functor preserves $\M$).
Dually, a factorization system also lifts to algebras, and even to the
Eilenberg-Moore category. 

Then, we will present different characterizations of minimality
(\autoref{figCharacterizations}) and then study properties of minimizations,
e.g.~under which criteria they exist and are unique, rediscovering the
respective proofs for reachability and observability for coalgebras in the
literature~\cite{amms13,Gumm03}. When combining the two minimization aspects, we
discuss under which criteria reachability and observability can be computed in
arbitrary order.

The goal of the present work is not only to show the connections between
existing minimality notions, but also to provide a series of basic results that
can be used when developing new minimization techniques or even new notions of
minimality.

\paragraph*{Related Work}
There is a series of
works~\cite{BidoitHK01,BezhanishviliKP12,BonchiBHPRS14,Rot16} that studies the
minimization of coalgebras by their duality to algebras. In those works, the
correspondence between observability in coalgebras and reachability in algebras
is used. For instance, Rot~\cite{Rot16} relates the final sequence (for
observability in coalgebras) with the initial sequence (for reachability in
algebras). In the present paper however, we consider both observability and
reachability on an abstract level that work for a general factorization system
and discuss their instance in coalgebras. The paper is based on Chapter 7 of
the author's PhD dissertation~\cite{Wissmann2020}.

If not included in the main text, detailed proofs of all results can be found in
the appendix%
\ifthenelse{\boolean{withappendix}}{%
}{ %
  of the full version%
}, both for the standard results recalled in the preliminaries and the new
results of the main sections.

\section{Preliminaries}
In the following, we assume basic knowledge of category theory~(cf.~standard
textbooks~\cite{joyofcats,awodey2010category}).

Given a diagram $D\colon \D\to \C$ (i.e.~a functor $D$ from a small category
$\D$), we denote its limit by $\lim D$ and colimit by $\colim D$ -- if they
exist. The limit projections, resp.~colimit injections, are denoted by
\[
  \pr_i\colon \lim D\to Di
  \qquad
  \inj_{i}\colon Di\to \colim D
  \qquad
  \text{for }i\in \D.
\]

\subsection{Coalgebra}

We model state-based systems as
coalgebras for an endofunctor $F\colon \C\to \C$ on a category $\C$:

\begin{definition}
  An \emph{$F$-coalgebra} (for an endofunctor $F\colon \C\to \C$) is a pair
  $(C,c)$ consisting of an object $C$ (of $\C$) and a morphism $c\colon C\to FC$
  (in $\C$). An $F$-coalgebra morphism $h\colon (C,c)\to (D,d)$ between
  $F$-coalgebras $(C,c)$ and $(D,d)$ is a morphism $h\colon C\to D$ with $d\cdot
  h = Fh\cdot c$ (see \autoref{fig:coalghom} on p.~\pageref{fig:coalghom} for the corresponding commuting diagram).
\end{definition}

Intuitively, the \emph{carrier} $C$ of a coalgebra $(C,c)$ is the state space
and the morphism $c\colon C\to FC$ sends states to
their possible next states. The functor of choice $F$ defines how these possible
next states $FC$ are structured.

\hackynewpage
\begin{example}
  Many well-known system-types can be phrased as coalgebras:
  \begin{itemize}
  \item Deterministic automata (without an explicit initial state) are
    coalgebras for the \Set-functor $FX = 2\times X^A$, where $A$ is the set of
    input symbols. In an $F$-coalgebra $(C,c)$, the first component of $c(x)$
    denotes the finality of the state $x\in C$ and the second component is the
    transition function $A\to C$ of the automaton.
  \item Labelled transition systems are coalgebras for the \Set-functor
    $FX=\Pow(A\times X)$ and the coalgebra morphisms preserve bisimilarity.
  \item Weighted systems with weights in a commutative monoid $(M,+,0)$ (and
    finite branching) are coalgebras for the \emph{monoid-valued functor}~\cite[Def.~5.1]{GummS01}
    \[
      M^{(X)} = \{\mu \colon X\to M\mid \mu(x) =0\text{ for all but finitely many
      }x\in X\}
    \]
    which sends a map $f\colon X\to Y$ to the map
    \[
    M^{(f)}\colon M^{(X)}\to M^{(Y)}
    \qquad M^{(f)}(\mu)(y) =\sum \{ \mu(x) \mid x\in X, f(x) = y\}
    \]
    In an $M^{(-)}$-coalgebra $(C,c)$, the transition weight from state $x\in C$
    to $y\in C$ is given by $c(x)(y) \in M$.
    E.g.~one obtains real-valued weighted systems as coalgebras for the functor
    $(\R,+,0)^{(-)}$.
  \item The \emph{bag} functor is defined
    by $\Bag X = (\N,+,0)^{(X)}$. Equivalently, $\Bag X$ is the set of finite
    multisets on $X$.
    Its coalgebras can be viewed as weighted systems or as transition systems
    in which there can be more than one transition between two states.
  \item A wide range of probabilistic and weighted systems can be obtained as
    coalgebras for respective distribution functors, see e.g.~Bartels
    \etal~\cite{BartelsSV04}.
  \end{itemize}
\end{example}

\begin{definition}
  The category of $F$-coalgebras and their morphisms is denoted by $\Coalg(F)$.
\end{definition}

Intuitively, the coalgebra morphisms preserve the behaviour of states:
\begin{definition}
  \label{defBehEq}
  In \Set, two states $x,y\in C$ in an $F$-coalgebra $(C,c)$ are
  \emph{behaviourally equivalent} if there is a coalgebra homomorphism
  $h\colon (C,c)\to (D,d)$ with $h(x) = h(y)$.
\end{definition}
\begin{example}
  \begin{itemize}
  \item For deterministic automata ($FX=2\times X^A$), states are behaviourally
    equivalent iff they accept the same language~\cite[Example 9.5]{Rutten00}.
    Indeed, for sufficiency, if two states $x,y$ in an $F$-coalgebra are
    identified by a coalgebra homomorphism, then one can show by induction over
    input words $w\in A^*$ that either both states or neither of them accepts
    $w$. For necessity, note that the map sending states to their semantics
    $C\to \Pow(A^*)$ is a coalgebra homomorphism.
  \item For labelled transition systems ($FX=\Pow(A\times X)$), states are
    behaviourally equivalent iff they are bisimilar~\cite{aczelmendler:89}.
  \item For weighted systems, i.e.~coalgebras for $M^{(-)}$, the coalgebraic
    behavioural equivalence captures weighted bisimilarity~\cite{KlinS13}.
  \item Further semantic notions can be modelled with coalgebras by changing the
    base category from $\C=\Set$ to the Eilenberg-Moore~\cite{TuriP97} or Kleisli
    category~\cite{HasuoJS06} of a monad, to nominal
    sets~\cite{KurzPSV13,MiliusSW16}, or to partially ordered sets~\cite{BalanK11}.
  \end{itemize}
\end{example}

The category of coalgebras inherits many properties from the base-category $\C$.
For instance, we have the following standard result:
\begin{corollary}
  \label{cor:coalg-colims}
  The forgetful functor $\Coalg(F)\to \C$ creates all colimits. That is, the
  colimit of a diagram $D\colon \D\to \Coalg(F)$ exists, if $U\cdot D\colon
  \D\to \C$ has a colimit, and moreover, there is a unique coalgebra structure on
  $\colim(UD)$ making it the colimit of $D$ and making the colimit injections of
  $\colim D$ coalgebra morphisms.
\end{corollary}

On the other hand, we do not necessarily have all limits in $\Coalg(F)$. If $F$
preserves a limit of a diagram $D\colon \D\to \C$, then the limit also exists in $\Coalg(F)$.

Coalgebras model systems with a transition structure, and pointed coalgebras
extend this by a notion of initial state:
\begin{definition}
  For an object $I\in \C$, an $I$-\emph{pointed} $F$-coalgebra $(C,c,i_C)$ is an
  $F$-coalgebra $(C,c)$ together with a morphism $i_C\colon I\to C$. A pointed
  coalgebra morphism $h\colon (C,c,i_C)\to (D,d,i_D)$ is a coalgebra morphism
  $h\colon (C,c)\to (D,d)$ that preserves the point: $i_D= h \cdot i_C$.

  The category of $I$-pointed $F$-coalgebras is denoted by $\Coalg_I(F)$.
\end{definition}

\begin{example} \label{exPointedDFA}
  For $I:= 1$ in $\Set$, a pointed coalgebra $(C,c,i_C)$ for $FX=2\times X^A$ is
  a deterministic automaton, where the initial state is given by the map
  $i_C\colon 1\to C$.
\end{example}

The point can also be understood as an algebraic flavour. In general, coalgebras
are dual to $F$-algebras in the following sense.

\begin{definition}
  An $F$-algebra (for a functor $F\colon \C\to\C$) is a morphism $a\colon FA\to
  A$, an algebra homomorphism $h\colon (A,a)\to (B,b)$ is a morphism $h\colon
  A\to B$ fulfilling $b\cdot Fh = h\cdot a$ (\autoref{fig:alghom}). The category of $F$-algebras is
  denoted by $\Alg(F)$.
\end{definition}

In other words, $\Alg(F) = \Coalg(F^\dual)^\dual$ for $F^\dual\colon
\C^\dual\to\C^\dual$. The $I$-pointed coalgebras thus are also algebras for the
constant $I$ functor.
Most of the results of the present paper also apply to algebras for a functor
$F\colon \C\to\C$.

\subsection{Factorization Systems}

The process of minimizing a system constructs a quotient or subobject of the
state space, where the notions of quotient and subobject respectively stem from
a factorization system in the category of interest. This generalizes the well-known
image factorization of a map into a surjective and an injective map:

\begin{figure}[t!]
  \begin{minipage}[b]{.49\textwidth}
    \begin{subfigure}{.5\textwidth}
      \begin{tikzcd}
        C \arrow{d}{h}\arrow{r}{c}& FC \arrow{d}[overlay]{Fh}\\
        D \arrow{r}{d} & FD
      \end{tikzcd}
      \caption{$F$-Coalgebras}
      \label{fig:coalghom}
    \end{subfigure}%
    \begin{subfigure}{.5\textwidth}
      \begin{tikzcd}
        FA
        \arrow{r}{a}
        \arrow{d}{Fh}
        & A
        \arrow{d}{h}
        \\
        FB \arrow{r}{b} & B
      \end{tikzcd}
      \caption{$F$-Algebras}
      \label{fig:alghom}
    \end{subfigure}%
    \caption{A homomorphism between $\ldots$}
  \end{minipage}%
  \hfill%
  \begin{minipage}[b]{.49\textwidth}
    \begin{subfigure}{.5\textwidth}
      \begin{tikzcd}[baseline=(A.base),row sep=2mm,column sep=4mm]
        |[alias=A]|
        A
        \arrow[->>,to path={|- (\tikztotarget) \tikztonodes},rounded corners]{dr}[pos=0.2]{e}
        \arrow[]{rr}[alias=f]{f}
        && B
        \\
        & \Im (f)
        \arrow[>->,to path={-| (\tikztotarget) \tikztonodes},rounded corners]{ur}[pos=0.8]{m}
      \end{tikzcd}
      \\[2mm] %
      \caption{Factorization}
      \label{fig:im}
    \end{subfigure}%
    \begin{subfigure}{.5\textwidth}
      \begin{tikzcd}[baseline=(A.base)]
        |[alias=A]|
        A
        \arrow[->>]{r}{e}
        \arrow{d}[swap]{f}
        & B
        \arrow{d}{g}
        \arrow[dashed]{dl}[description]{\exists !d}
        \\
        C
        \arrow[>->]{r}{m}
        & D
      \end{tikzcd}
      \caption{Diagonal fill-in}
      \label{fig:diagonal}
    \end{subfigure}%
    \caption{Diagrams for \autoref{D:factSystem}}
  \end{minipage}
\end{figure}

\begin{notheorembrackets}
\begin{definition}[{\cite[Definition 14.1]{joyofcats}}] \label{D:factSystem}
  Given classes of morphisms $\E$ and $\M$ in $\C$, we say that $\C$ has an
  \emph{$(\E,\M)$-factorization system}\index{EM@$(\E,\M)$}\index{factorization system} provided that:
  \begin{enumerate}
  \item $\E$ and $\M$ are closed under composition with isomorphisms.
  \item
    Every morphism $f\colon A\to B$ in $\C$ has a factorization $f =
    m\cdot e$ with $e\in \E$ and $m\in \M$ (\autoref{fig:im}). We write $\Im(f)$
    for the intermediate object,
    $\twoheadrightarrow$ for morphisms 
    $e\in \E$, and $\monoto$ for morphisms $m\in \M$.
  \item\label{diagonalization}
    For each commutative square $g\cdot e = m\cdot f$ with $m\in \M$ and $e\in \E$,
    there exists a unique diagonal fill-in $d$ with $m\cdot d=g$ and $d\cdot e =
    f$ (\autoref{fig:diagonal}).
  \end{enumerate}
\end{definition}
\end{notheorembrackets}

\begin{example}
  In \Set, we have an $(\Epi,\Mono)$-factorization system
  where $\Epi$ is the class of surjective maps, and $\Mono$ the class of injective maps.
  The image of a map $f\colon A\to B$ is given by
  \[
    \Im(f) = \{b\in B \mid \text{there exists $a\in A$ with $f(a) = b$}\}.
  \]
  canonically yielding maps $e\colon A\twoheadrightarrow \Im(f)$ and $m\colon \Im(f)\monoto B$.
  Note that one can also regard $\Im(f)$ as a set of equivalence classes of $A$:
  \[
    \Im(f) \cong \big\{ \{a'\in A\mid f(a') = f(a)\} \,\mid a\in A \big\}.
  \]
  Intuitively, the diagonal fill-in property
  (\itemref{D:factSystem}{diagonalization}, also called \emph{diagonal lifting})
  provides a way of defining a map $d$ on equivalence classes (given by the
  surjective map at the top) and with a restricted codomain (given by the
  injective map at the bottom).
\end{example}

\begin{example}
  \label{trivFact}
  In general, the elements of $\E$ are not necessarily epimorphisms and the
  elements of $\M$ are not necessarily monomorphisms. In particular, every
  category has an $(\E,\M)$-factorization system with $\E := \Iso$ being the class
  of isomorphisms and $\M := \Mor$ being the class of all morphisms (and also
  vice-versa).
\end{example}

\begin{definition}
  An $(\E,\M)$-factorization system is called \emph{proper} if $\E\subseteq
  \Epi$ and $\M\subseteq \Mono$.
\end{definition}
These two conditions of properness are independent. In fact, $\M\subseteq \Mono$ is
equivalent to every split-epimorphism being in $\E$~\cite[Prop.~14.11]{joyofcats}.
In the literature, it is often required that the factorization system is proper,
and in fact a proper factorization system arises in complete or cocomplete
categories:

\begin{example} \label{completeCatFactorizations} Every complete category has
  a $(\StrongEpi, \Mono)$-factorization system~\cite[Thm.~14.17 and
  14C(d)]{joyofcats} and also an $(\Epi, \StrongMono)$-factorization system
  \cite[Thm.~14.19, and 14C(f)]{joyofcats}. By duality, every cocomplete
  category has so as well.
\end{example}

\begin{remark} \label{rem:EM}
  $(\E,\M)$-factorization systems have many properties known from surjective
  and injective maps on $\Set$~\cite[Chp.~14]{joyofcats}:
  \begin{enumerate}
  \item\label{rem:EM:iso} $\E\cap \M$ is the class of isomorphisms of $\C$.
  \item\label{rem:EM:3} If $f\cdot g\in \M$ and $f\in \M$, then $g\in \M$. 
    If $\M\subseteq \Mono$, then $f\cdot g\in \M$ implies $g\in \M$.
  \item $\E$ and $\M$ are respectively closed under composition.
  \item\label{rem:EM:pullback} $\M$ is stable under pullbacks, $\E$ is stable under pushouts.
  \end{enumerate}
\end{remark}

The stability generalizes as follows to wide pullbacks and pushouts:
\begin{lemma}
  \label{EM:pullbackwide}
  $\M$ is stable under wide pullbacks: for a
    family $(f_i\colon A_i\to B)_{i\in I}$ and its wide pullback $(\pr_i\colon
    P\to A_i)_{i\in I}$, a projection $\pr_j\colon P\to A_j$ is in $\M$ if $f_i$
    is in $\M$ for all $i \in I\setminus \{j\}$.
\end{lemma}

A factorization system also provides notions of subobjects and quotients,
generalizing the notions of subset and quotient sets:

\begin{definition}
  \label{D:subobjectEM}
  For a class $\M$ of morphisms, an \emph{$\M$-subobject} of an object $X$ is 
  a pair $(S,s)$ where
  $s\colon S\monoto X$ is in $\M$.
  Two $\M$-subobjects $(s,S)$, $(s',S')$ are called \emph{isomorphic}
  if there is an isomorphism $\phi\colon S\to S'$ with $\phi\cdot s = s'$.
  We write $(s,S)\le (s', S')$ if there is a morphism $h\colon S\to S'$ with $s'\cdot h = s$.
  Dually, an \emph{$\E$-quotient} of $X$ is pair $(Q,q)$ for a
  morphism $q\colon X\epito Q$ ($q\in \E$).
  If $(\E,\M)$ is fixed from the context, we simply speak of
  subobjects and quotients.
\end{definition}
For subobjects, it is often required that $\M$ is a class of
monomorphisms~\cite[Def.~7.77]{joyofcats}, but many of the results in the
present work hold without this assumption. If $\M$ is so, then the subobjects of
a given object $X$ form a preordered class. Moreover, the subobjects form a
preordered set iff $\C$ is $\M$-wellpowered. This is in fact
the definition: $\C$ is $\M$-wellpowered if for each $X\in \C$ there is (up to
isomorphism) only a set of $\M$-subobjects. On $\Set$, the isomorphism classes
of \text{($\Mono$-)}sub\-ob\-jects of $X$ correspond to subsets of $X$ and the
isomorphism classes of ($\Epi$-)quotients of $X$ correspond to partitions on
$X$.

If $(\E,\M)$ forms a factorization system, then its axioms provide us with
methods to construct and work with subobjects and quotients, e.g.~the image
factorization means that for every morphism, we obtain a quotient on its domain
and a subobject of its codomain. The minimization of coalgebras amounts to the
construction of certain subobjects or quotients with
respect to a suitable factorization system in the category of coalgebras $\Coalg(F)$.

\section{Factorization System for Coalgebras}
\label{secCoalgFact}

If we have an $(\E,\M)$-factorization system on the base category $\C$ on which
we consider coalgebras for $F\colon \C\to\C$, then it is natural
to consider coalgebra morphisms whose underlying $\C$-morphism is in $\E$, resp.~$\M$:
\begin{definition}
  Given a class of $\C$-morphisms $\E$, we say that an $F$-coalgebra morphism
  $h\colon (C,c)\to (D,d)$ is \emph{$\E$-carried} if $h\colon C\to D$ is in $\E$.
\end{definition}

This induces the standard notions of subcoalgebra and quotient coalgebras as
instances of $\M$-subobjects and $\E$-quotients in $\Coalg(F)$: an
$\M$-subcoalgebra of $(C,c)$ is an ($\M$-carried)-subobject of $(C,c)$ (in
$\Coalg(F)$), i.e.\ is represented by an $\M$-carried homomorphism $m\colon (S,s)
\monoto (C,c)$.
Likewise, a \emph{quotient}\index{quotient coalgebra} of a coalgebra $(C,c)$ is
an ($\E$-carried)-quotient of $(C,c)$ (in $\Coalg(F)$), i.e.\
is represented by a coalgebra morphism $e\colon (C,c) \epito (Q,q)$ carried by an
epimorphism.

Note that for the case where $\M$ is the class of monomorphisms, the
monomorphisms in $\Coalg(F)$ coincide with the $\Mono$-carried homomorphisms
only under additional assumptions:
\begin{lemma}
  \label{coalgMono}
  If weak kernel pairs exist in $\C$ and are preserved by $F\colon \C\to \C$, then the
  monomorphisms in $\Coalg(F)$ are precisely the $\Mono$-carried coalgebra
  homomorphisms.
\end{lemma}
This is a commonly known criterion, and Gumm and Schröder~\cite[Example 3.5]{GS05} present a functor not preserving kernel pairs and a monic coalgebra homomorphism that is not carried by a monomorphism.

For the construction of quotient coalgebras and subcoalgebras, it is handy to
have the factorization system directly in $\Coalg(F)$. It is a standard
result that the image factorization of homomorphisms lifts (see e.g.~\cite[Lemma
2.5]{lffiac}). Under assumptions on $\E$ and $\M$, Kurz shows that the
factorization system lifts to \Coalg(F)~\cite[Theorem 1.3.7]{kurzPhd} (and to
other categories with a forgetful functor to the base category $\C$).

In fact, the factorization system always lifts to $\Coalg(F)$ under the
condition that $F$ preserves $\M$. By this condition, we mean that $m\in \M$
implies $Fm\in \M$.
\begin{lemma}
  \label{coalgfactor}
  If $F\colon \C\to\C$ preserves $\M$, then the $(\E,\M)$-factorization system
  lifts from $\C$ to an ($\E$-carried, $\M$-carried)-factorization system in
  $\Coalg(F)$. The factorization of an $F$-coalgebra homomorphism $h\colon (C,c)\to (D,d)$
  is given by that of the underlying morphism $h\colon C\to D$.
\end{lemma}
\begin{proof}
  We verify \autoref{D:factSystem}:
  \begin{enumerate}
  \item The $\E$- and $\M$-carried morphisms are closed under composition
    with isomorphisms, respectively.

  \item
    Given an $F$-coalgebra morphism $f\colon (A,a)\to (B,b)$,
  consider its factorization $f=m\cdot e$ in $\C$.
  Since $F$ preserves $\M$, we have $Fm\in \M$ and thus can apply the
  diagonal fill-in property (\itemref{D:factSystem}{diagonalization})
  to the coalgebra morphism square of $f$ (\autoref{coalgfactorStep2}). This defines a unique
  coalgebra structure $d$ on $\Im(f)$ making $e$ and $m$ coalgebra morphisms.
    \begin{figure}[t]
      \begin{subfigure}[b]{.36\textwidth}
        \begin{tikzcd}[baseline={([yshift=-2mm]f.base)}]
          |[alias=A]|
          A
          \arrow[shiftarr={yshift=6mm}]{rr}[alias=f]{f}
          \arrow[->>]{r}{e}
          \arrow{d}{a}
          & \Im(f)
          \arrow[>->]{r}{m}
          \arrow[dashed]{d}{\exists! d}
          & B
          \arrow{d}{b}
          \\
          FA
          \arrow[->]{r}{Fe}
          \arrow[shiftarr={yshift=-6mm}]{rr}[swap]{Ff}
          & F\Im(f)
          \arrow[>->]{r}{Fm}
          & FB
        \end{tikzcd}
        \caption{Factorization}
        \label{coalgfactorStep2}
      \end{subfigure}
      \hfill
      \begin{subfigure}[b]{.26\textwidth}
        \begin{tikzcd}[baseline=(A.base)]
          |[alias=A]|
          (A,a)
          \arrow[->>]{r}{e}
          \arrow{d}[swap]{f}
          & (B,b)
          \arrow{d}{g}
          \\
          (C,c)
          \arrow[>->]{r}{m}
          \arrow[draw=none,shiftarr={yshift=-6mm}]{r}[swap]{\phantom{Ff}}
          & (D,d)
        \end{tikzcd}
        \caption{Premise in $\Coalg(F)$}
        \label{coalgfactorAssump}
      \end{subfigure}
      \hfill
      \begin{subfigure}[b]{.26\textwidth}
        \begin{tikzcd}[sep=9mm,baseline=(A.base)]
          |[alias=A]|
          A
          \arrow[->>]{r}{e}
          \arrow{d}[swap]{f}
          & B
          \arrow{d}{g}
          \arrow[dashed]{dl}[description]{\exists !h}
          \\
          C
          \arrow[>->]{r}{m}
          \arrow[draw=none,shiftarr={yshift=-5mm}]{r}[swap]{\phantom{Ff}}
          & D
        \end{tikzcd}
        \caption{Fill-in in $\C$}
        \label{coalgfactorBase}
      \end{subfigure}
      \caption{Diagrams for the proof of \autoref{coalgfactor}}
    \end{figure}

  \item 
  In order to check the diagonal-lifting property of the
  ($\E$-carried, $\M$-carried)-factorization system, consider a commutative
  square $g\cdot e = m\cdot f$ in $\Coalg(F)$ with $m\in \M$, $e\in \E$ as
  depicted in \autoref{coalgfactorAssump}.
  In $\C$, there exists a unique $h\colon B\to C$ with
  $h\cdot e = f$ and $m\cdot h =g$ (\autoref{coalgfactorBase}).
  We only need to prove that $h\colon B\to C$ is a coalgebra homomorphism
  $(B,b)\to (C,c)$, i.e.~that $c\cdot h = Fh\cdot b$. We prove this equality by
  showing that both $c\cdot h$ and $Fh\cdot b$ are diagonals in a commutative
  square of the form of \itemref{D:factSystem}{diagonalization}.
  Indeed, we have the commutative squares:
  
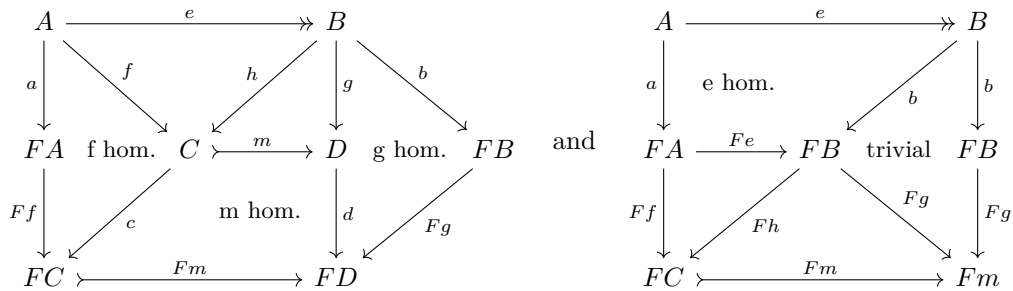
\[
    \begin{tikzcd}[sep=12mm]
      A
      \arrow{d}[swap]{a}
      \arrow{dr}{f}
      \arrow[->>]{rr}{e}
      &
      & B
      \arrow{dl}[swap]{h}
      \arrow{d}{g}
      \arrow{dr}{b}
      \\
      FA
      \arrow{d}[swap]{Ff}
      \descto{r}{f\text{ hom.}}
      & C
      \arrow[>->]{r}{m}
      \arrow{dl}{c}
      \descto{dr}{m\text{ hom.}}
      & D
      \descto{r}{g\text{ hom.}}
      \arrow{d}{d}
      & FB
      \arrow{dl}{Fg}
      \\
      FC
      \arrow[>->]{rr}{Fm}
      & & FD
    \end{tikzcd}
    \quad\text{and}\quad
    \begin{tikzcd}[sep=12mm]
      A
      \arrow[->>]{rr}{e}
      \arrow{d}[swap]{a}
      \descto{dr}{e\text{ hom.}}
      && B
      \arrow{dl}{b}
      \arrow{d}{b}
      \\
      FA
      \arrow{d}[swap]{Ff}
      \arrow{r}{Fe}
      & FB
      \arrow{dl}{Fh}
      \arrow{dr}{Fg}
      \descto{r}{\text{trivial}}
      & FB
      \arrow{d}{Fg}
      \\
      FC
      \arrow[>->]{rr}{Fm}
      & & Fm
    \end{tikzcd}
  \]
  By the uniqueness of the diagonal in \itemref{D:factSystem}{diagonalization}, $c\cdot h =
  Fh\cdot b$.
  \qedhere
  \end{enumerate}
\end{proof}

\begin{remark}
  The condition that $F$ preserves $\M$ is commonly met. For $\Set$ and $\M$
  being the class of injective maps, it can be assumed wlog for coalgebraic
  purposes that $F$ preserves injective maps: every set functor preserves
  injective maps with non-empty domain and only needs to be modified on
  $\emptyset$ in order to preserve all injective maps~\cite{trnkova71}.
  The resulting functor has an isomorphic category of coalgebras.
\end{remark}

We have the dual result for $F$-algebras:
\begin{lemma}
  \label{algfactor}
  If $F\colon \C\to\C$ preserves $\E$, then the $(\E,\M)$-factorization system
  lifts from $\C$ to $\Alg(F)$.
\end{lemma}
\begin{proof}
  We have an $(\M,\E)$-factorization system in $\C^\dual$. By
  \autoref{coalgfactor}, this factorization system lifts to $\Coalg(F^\dual)$
  since $F^\dual\colon \C^\dual\to \C^\dual$ preserves $\E$. Thus, we have an
  ($\E$-carried, $\M$-carried)-factorization system in $\Alg(F) = \Coalg(F^\dual)^\dual$.
\end{proof}

This lifting result even holds for Eilenberg-Moore algebras for a monad $T\colon
\C\to\C$. The Eilenberg-Moore category of a monad $T$ is a full subcategory of
$\Alg(T)$ containing those algebras that interact coherently with the structure
of the monad $T$, see~e.g.~Awodey~\cite{awodey2010category} for details:
\begin{lemma}
  \label{monadfactor}
  If a monad $T\colon \C\to\C$ preserves $\E$, then the $(\E,\M)$-factorization system
  lifts from $\C$ to the Eilenberg-Moore category of $T$.
\end{lemma}

The factorization system lifts further also to pointed coalgebras:

\begin{lemma}
  If $F\colon \C\to\C$ preserves $\M$, then the $(\E,\M)$-factorization system
  lifts from $\C$ to $\Coalg_I(F)$.
\end{lemma}
\begin{proof}
  A combination of \autoref{coalgfactor} and \ref{algfactor}, using that the
  constant functor preserves $\E$-morphisms.
  \takeout{}
\end{proof}

\section{Minimality in a Category}
\label{secMinObj}
Having seen multiple categories with an $(\E,\M)$-factorization system, we can now
define the minimality of objects abstractly.
\begin{definition}
  Given a category $\K$ with an $(\E,\M)$-factorization system, an object $C$ of
  $\K$ is called \emph{$\M$-minimal}\index{minimal@$\M$-minimal} if every
  morphism $h\colon D\monoto C$ in $\M$ is an isomorphism.
\end{definition}
\begin{remark}
  Every $(\E,\M$)-factorization system on $\K$ is an $(\M,\E)$-factorization
  system on $\K^\dual$, and thus induces a dual notion of
  $\E$-minimality: an object $C$ of $\K$ is called \emph{$\E$-minimal} if every $h\colon
  C\epito D$ in $\E$ is an isomorphism.
\end{remark}

In the following, $\K$ will denote
the category in which we consider the minimal objects, e.g.~a category of
coalgebras for a functor $F\colon \C\to \C$.
\begin{assumption}
  \label{assDualFactor}
  In the following, assume that the category $\K$ has an $(\E,\M)$-fac\-tor\-iza\-tion
  system. Whenever we consider a category of coalgebras for a functor $F\colon
  \C\to \C$ in the following, we achieve this by assuming that $\C$ has an
  $(\E,\M)$-factorization system and that $F$ preserves~$\M$.
\end{assumption}
The leading examples of the minimality notion in the present work are the following
two instances in coalgebras:
\begin{instance}
  For $\K := \Coalg_I(F)$, the ($\M$-carried-)minimal objects are the \emph{reachable}
  coalgebras, as introduced by~\Adamek\ \etal~\cite{amms13}.
  Concretely, an $I$-pointed $F$-coalgebra $(C,c,i_C)$ is reachable if it has no (proper)
  pointed subcoalgebra, equivalently, if every $\M$-carried coalgebra morphism
  $m\colon (S,s,i_S)\monoto (C,c,i_C)$ is necessarily an isomorphism~\cite{amms13}.

  In \Set, this corresponds to the usual notion of reachability: if a state
  $x\in C$ is contained in a subcoalgebra $m\colon (S,s,i_S)\monoto (C,c,i_C)$,
  then all successors of $x$ need to be contained in the subcoalgebra as well,
  since $m$ is a coalgebra homomorphism. Moreover, the subcoalgebra has to
  contain the point $i_C\colon I\to C$, and thus also all
  its successors, and in total all states reachable from $i_C$ in finitely many
  steps. Hence, $(C,c,i_C)$ is reachable if it is not possible to omit any
  state in a pointed subcoalgebra $(S,s,i_S)$, i.e.~if any such injective $m$ is
  a bijection.
\end{instance}

\begin{instance}
  For $\K := \Coalg(F)^\dual$, the ($\E$-carried-)minimal objects are called
  \emph{simple} coalgebras, as mentioned~by Gumm~\cite{Gumm03}. Usually, a
  simple coalgebra is defined as a coalgebra that does not have any proper
  quotient~\cite{WissmannDMS19}.\footnote{Gumm~\cite[p.~34]{Gumm03} defines a simple coalgebra as the
    quotient of a coalgebra on $\Set$ modulo behavioural equivalence.}

  In \Set, a coalgebra is simple iff all states have different behaviour --
  this characterization follows
  directly the following equivalent characterization of minimal objects as we
  will see in \autoref{instSimpleBehEq}:
\end{instance}

\begin{lemma}
  \label{minimalInE}
  An object $C$ in $\K$ is $\M$-minimal iff every $h\colon D\to C$ is in $\E$.
\end{lemma}
\begin{proof}
  In the \singlequote{if} direction, consider some $\M$-morphism $h\colon D\to
  C$. By the assumption, $h$ is also in $\E$ and thus an isomorphism.
  In the \singlequote{only if} direction, take some morphism $h\colon D\to C$
  and consider its $(\E,\M)$-factorization $e\colon D \epito \Im(h)$ and
  $m\colon \Im(h)\monoto C$ with $h=m\cdot e$. Since $C$ is $\M$-minimal, $m$ is
  an isomorphism and thus $h=m\cdot e$ is in $\E$.
\end{proof}
\begin{instance}
  \label{instSimpleBehEq}
  For $\K:=\Coalg(F)^\dual$, an $F$-coalgebra $(C,c)$ is simple iff every
  $F$-coalgebra homomorphism $h\colon (C,c)\to (D,d)$ is $\M$-carried.
\end{instance}

  In \Set, this equivalence shows that the simple coalgebras are precisely those
  coalgebras for which behavioural equivalence coincides with behavioural
  equivalence:
  \begin{itemize}
  \item If states $x,y\in C$ are behaviourally equivalent, then there is some
    $h\colon (C,c)\to (D,d)$ with $h(x)=h(y)$. By \autoref{instSimpleBehEq},
    $h$ must be injective and thus $x=y$.

  \item Conversely, if all states in $(C,c)$ have different behaviour, then
    every $h\colon (C,c)\to (D,d)$ is necessarily injective by
    \autoref{defBehEq}. Thus, by \autoref{instSimpleBehEq}, $(C,c)$ is simple.
  \end{itemize}

Gumm already noted that in $\Set$, every outgoing coalgebra morphism from a
simple coalgebra is injective~\cite[Hilfssatz~3.6.3]{Gumm03} -- but it was not
used to characterize simplicity. If $\E$, resp.~$\M$, happens to be the class of
epimorphisms, resp.~monomorphisms, another characterization of minimality exists:

\begin{lemma}
  \label{minimalSubterminal}
  Assume $\E=\Epi$ and weak equalizers in $\K$, then $X$ is $\M$-minimal iff
  there is at most one morphism $u\colon C\to D$ for every $D\in \K$.

  Dually, given $\M=\Mono$ and weak coequalizers, $C$ is $\E$-minimal iff $C$ is
  \emph{subterminal}, that is, iff there is a most one $u\colon D\to C$ for
  every $D\in \K$.
\end{lemma}
The name \emph{subterminal} stems from the fact that if $\K$ has a terminal
object, its subobjects are the subterminal objects.

\begin{instance}
  For $\K := \Coalg(F)^\dual$, assume $\M=\Mono$ and that $F$ preserves weak
  kernel pairs and that the base category $\C$ has coequalizers. Hence, the monomorphisms in $\Coalg(F)$ are precisely the
  $\Mono$-carried homomorphisms~(\autoref{coalgMono}) and the assumption of
  \autoref{minimalSubterminal} is met. Consequently, the simple
  coalgebras are precisely the \emph{subterminal} coalgebras. If the final
  coalgebra exists, then its subcoalgebras are precisely the simple coalgebras.
  For a non-example,
  Gumm and Schröder~\cite[Example 3.5]{GS05} provide a functor not 
  preserving weak kernel pairs and a subterminal coalgebra that is not simple.
\end{instance}

\begin{figure}
  \begin{tabular}{@{}c@{\ }lc@{}c@{\,}l@{}}
    \cmidrule[\heavyrulewidth]{1-2}
    \cmidrule[\heavyrulewidth]{4-5}
    \multicolumn{2}{@{}l}{$X$ is $\M$-minimal}
    &
      \hspace*{5mm} %
    & \multicolumn{2}{@{}l}{$X$ is $\E$-minimal}
    \\
    \cmidrule{1-2} \cmidrule{4-5}
    \ensuremath{\Leftrightarrow} &
    every $Y\monoto X$ is an isomorphism
    && \ensuremath{\Leftrightarrow} &every $X\epito Y$ is an isomorphism
    \\
    \cmidrule{1-2} \cmidrule{4-5}
    \ensuremath{\Leftrightarrow} &
    every $Y\to X$ is in $\E$
    &&\ensuremath{\Leftrightarrow} & every $X\to Y$ is in $\M$
    \\
    \cmidrule{1-2} \cmidrule{4-5}
    &
    \makecell[l]{
      \textshaded{if $\E=\Epi$ and $\K$ has weak equalizers:} \\
    }
    &&& \makecell[l]{
       \textshaded{if $\M=\Mono$ and $\K$ has weak coequalizers:}
       }
    \\
    \ensuremath{\Leftrightarrow} &
      all parallel $X\rightrightarrows Y$ equal
    && \ensuremath{\Leftrightarrow} &
      all parallel $Y\rightrightarrows X$ equal
      (\textqt{X subterminal})
    \\
    \cmidrule[\heavyrulewidth]{1-2}
    \cmidrule[\heavyrulewidth]{4-5}
  \end{tabular}
  \caption{Equivalent characterizations of $\M$-minimality and
    $\E$-minimality in a category $\K$}
  \label{figCharacterizations}
\end{figure}

We have now established a series of equivalent characterizations of minimality
(\autoref{figCharacterizations}) and will now discuss how to construct minimal objects.
This process of minimization -- i.e.~of constructing the reachable part or the
simple quotient of a coalgebra -- is abstracted  as follows:

\begin{definition}
  An \emph{$\M$-minimization}\index{minimization@$\M$-minimization} of $C\in \K$
  is a morphism $m\colon D\monoto C$ in $\M$ where $D$ is $\M$-minimal.
\end{definition}

In fact, we will show that an $\M$-minimization is unique, so we can speak of
\emph{the} $\M$-minimization.

\begin{instance}
  The task of finding an $\M$-minimization of a given $C\in \K$ instantiates to
  the standard minimization tasks on coalgebras:
  \begin{itemize}
  \item For $\K:=\Coalg_I(F)$, an $\M$-minimization of a given pointed coalgebra
    $(C,c,i_C)$ is called its \emph{reachable subcoalgebra}~\cite{amms13}. This is a
    subcoalgebra obtained by removing all unreachable states. The explicit
    definition is: the reachable subcoalgebra of $(C,c,i_C)$ is a (pointed)
    subcoalgebra $h\colon (R,r,i_R)\monoto (C,c,i_C)$ where $(R,r,i_R)$ itself has no proper
    (pointed) subcoalgebras.
  \item For $\K:= \Coalg(F)^\dual$, an $\E$-minimization of a given coalgebra
    $(C,c)$ is called the \emph{simple quotient} of $(C,c)$~\cite{Gumm03}. The
    explicit definition is: the simple quotient of $(C,c)$ is a quotient
    $h\colon (C,c)\epito (Q,q)$ where $(Q,q)$ itself has no proper quotient
    coalgebra.

    In \Set, this is a quotient in which all behaviourally equivalent states are
    identified, in other words, the simple quotient of $(C,c)$ is the unique
    coalgebra structure on $C/\mathord{\sim}$. Examples of simple quotients can
    be found in \autoref{figSimpleQuot}. Since all states in the codomain of the
    surjective homomorphisms are behaviourally inequivalent, the respective
    codomains are simple.
  \end{itemize}
\end{instance}

\begin{figure}
  \begin{subfigure}{.5\textwidth}
    \begin{tikzpicture}[coalgebra,baseline=(q0.base)]
      \node[state] (q0) {$\bullet$};
      \node[state] (q1) at (1.3,0.5) {$\bullet$};
      \node[state] (q2) at (1.3,-0.5) {$\bullet$};
      \draw[transition] ([xshift=-3mm]q0.west) to (q0.west);
      \draw[transition] (q0) to (q1);
      \draw[transition] (q0) to (q2);
      \draw[transition] (q1) to (q2);
      \draw[transition,loop right] (q1) to (q1);
    \end{tikzpicture}
    \begin{tikzcd}
      {} \arrow[->>]{r} & {}
    \end{tikzcd}
    \begin{tikzpicture}[coalgebra,baseline=(q0.base)]
      \node[state] (q0) {$\bullet$};
      \node[state] (q1) at (1.3,0) {$\bullet$};
      \draw[transition] ([xshift=-3mm]q0.west) to (q0.west);
      \draw[transition] (q0) to (q1);
      \draw[transition,loop above] (q0) to (q0);
    \end{tikzpicture}
    \subcaption{For $FX=\Pow X$}
  \end{subfigure}%
  \hfill
  \begin{subfigure}{.48\textwidth}
        \begin{tikzpicture}[coalgebra,baseline=(q0.base)]
          \node[state] (q0) {$\bullet$};
          \node[state] (q1) at (1.3,0.5) {$\bullet$};
          \node[state] (q2) at (1.3,-0.5) {$\bullet$};
          \draw[transition] ([xshift=-3mm]q0.west) to (q0.west);
          \draw[transition] (q0) to node[above left] {4}(q1);
          \draw[transition] (q0) to node[below left] {-7} (q2);
          \draw[transition] (q1) to node[right] {5}(q2);
          \draw[transition,loop right,overlay] (q2) to node[right] {5}(q2);
        \end{tikzpicture}
        \begin{tikzcd}
          {} \arrow[->>]{r}{}
          & {}
        \end{tikzcd}
        \begin{tikzpicture}[coalgebra,baseline=(q0.base)]
          \node[state] (q0) {$\bullet$};
          \node[state] (q1) at (1,0) {$\bullet$};
          \draw[transition] ([xshift=-3mm]q0.west) to (q0.west);
          \draw[transition] (q0) to node[above] {-3} (q1);
          \draw[transition,loop right] (q1) to node[right] {5}(q1);
        \end{tikzpicture}
    \subcaption{For $FX=(\R,+,0)^{(X)}$}
  \end{subfigure}
  \caption{Examples of simple quotients in $F$-coalgebras}
  \label{figSimpleQuot}
\end{figure}
  
\begin{samepage}
\begin{example} \label{trivMinimal}
  For the trivial factorization systems (\autoref{trivFact}), we have:
  \begin{itemize}
  \item For the $(\Iso,\Mor)$-factorization system, the $\Iso$-minimization of
    an object $X$ is $X$ itself.
  \item For the $(\Mor,\Iso)$-factorization system on category, if a strict
    initial object $0$ exist, then it is the $\Mor$-minimization of every $X\in
    \C$. Recall that an initial object $0$ is called \emph{strict} if every
    morphism with codomain $0$ is an isomorphism.
  \end{itemize}
\end{example}
\end{samepage}

It is well-defined to speak of \emph{the} $\M$-minimization of an object $C$,
because it is unique:

\begin{lemma}\label{MminLeast}
  Consider $h\colon M\to C$ with $\M$-minimal $M$ and an $\M$-subobject
  $s\colon S\monoto C$. The pullback of $s$ along $h$ exists iff $h$ factors
  uniquely through $s$, that is, iff there is a unique $u\colon M\to S$ with $s\cdot u = h$.
\end{lemma}
\[
  \begin{tikzcd}[column sep=5pt]
    &C
    \\
    M
    \arrow[->]{ur}{h}
    \arrow[->,dashed]{rr}{\exists!u}
    &&
    S
    \arrow[>->]{ul}[swap]{\forall s\in \M}
  \end{tikzcd}
  ~~\text{(in $\K$)}
  \qquad
  \begin{tikzcd}[column sep=5pt]
    &C
    \\
    M
    \arrow[<-]{ur}{h}
    \arrow[<-,dashed]{rr}{\exists!u}
    &&
    Q
    \arrow[<<-]{ul}[swap]{\forall q\in \E}
  \end{tikzcd}
  ~~\text{(in $\K^\dual$)}
\]
\begin{proof}
  In the `if'-direction, let $d\colon M\to S$ be the unique morphism with
  $s\cdot d = h$. The pullback is simply given by $M$ itself with
  projections $\id_M\colon M\to M$ and $d\colon M\to S$.
  To verify its universal
  property, consider $e\colon E\to M$, $f\colon E\to S$ with $h\cdot e = s\cdot
  f$ (\autoref{fig:MminIf}).
  Since $M$ is $\M$-minimal, $e\colon E\to M$ is in $\E$ (\autoref{minimalInE}).
  Thus, we can apply the diagonal lifting property to $h\cdot e = s\cdot f$
  yielding a diagonal $u$ with $s\cdot u=h$ and $u\cdot e = f$.
  Thus, $d = u$ and $d\cdot e = f$, showing that $e\colon (E,e,f)\to (M,\id_M,d)$ is the
  mediating cone morphism. Its uniqueness is clear because $\id_M$ is isomorphic.
  \begin{figure}[h!]
    \begin{subfigure}[b]{.3\textwidth}
      \begin{tikzcd}
        M
        \arrow{r}{h}
        \arrow[dashed]{dr}[description]{u}
        & C
        \\
        E
        \arrow[->>]{u}{e}
        \arrow{r}{f}
        & S
        \arrow[>->]{u}[swap]{s}
      \end{tikzcd}
      \caption{`If'-direction}
      \label{fig:MminIf}
    \end{subfigure}%
    \begin{subfigure}[b]{.3\textwidth}%
      \begin{tikzcd}
        P
        \arrow[>->]{r}{\phi}
        \arrow[->]{d}[swap]{d}
        \pullbackangle{-45}
        &
        M
        \arrow[->]{d}{h}
        \\
        S
        \arrow[>->]{r}{s}
        & C
      \end{tikzcd}
      \caption{Premise of `only if'}
      \label{fig:MminPullback}
    \end{subfigure}
    \begin{subfigure}[b]{.3\textwidth}
      \begin{tikzcd}[row sep=0mm,column sep=7mm]
        |[yshift=2mm]|
        M
        \arrow[bend left=20]{rr}{\id_M}
        \arrow[bend right=10]{dr}[swap]{v}
        \arrow[dashed]{r}{u}
        &
        P
        \arrow{r}[pos=0.2]{\phi}
        \arrow{d}[swap]{d}
        \pullbackangle{-45}
        &
        M
        \arrow[->]{d}{h}
        \\[4mm]
        & S
        \arrow[>->]{r}{s}
        & C
      \end{tikzcd}
      \caption{Uniqueness for `only if'}
      \label{fig:MminUMP}
    \end{subfigure}
    \caption{Diagrams for the proof of \autoref{MminLeast}}
  \end{figure}
  
  In the `only if'-direction, consider the pullback $(P,\phi,d)$ (\autoref{fig:MminPullback}).
  Since $\M$-morphisms are stable under pullback
  (\itemref{rem:EM}{rem:EM:pullback}), $\phi$ is in
  $\M$, too. By the minimality of $M$, the $\M$-morphism $\phi$ is an isomorphism
  and we have $d\cdot \phi^{-1}\colon M\to S$.

  In order to see that $d\cdot \phi^{-1}$ is indeed the unique morphism $M\to S$ making the
  triangle commute, consider an arbitrary $v\colon M\to S$ with $s\cdot v = h =
  h\cdot \id_M$.
  Thus, $M$ is a competing cone for the pullback $P$ and thus induces a morphism
  $u\colon M\to P$ with $d\cdot u = v$ and $\phi\cdot u =\id_M$ (\autoref{fig:MminUMP}). Since $\phi$ is
  an isomorphism, we have $u=\phi^{-1}$ and thus $v=d\cdot \phi^{-1}$ as desired.
\end{proof}
\begin{corollary}
  \label{uniqueMinimization}
  If $\K$ has pullbacks of $\M$-morphisms along $\M$-morphisms and if there is an
  $\M$-minimization $M$ of $C$, then $M$ is the least $\M$-subobject
  of $C$ and it is unique (up to unique isomorphism).
\end{corollary}%
\begin{proof}
  Consider \autoref{MminLeast} first for $h\in \M$ and then also with
  $S$ being $\M$-minimal.
\end{proof}

\begin{instance} Not only the result but also the proof instantiates to
  the uniqueness results in the instances of reachable subcoalgebras and simple quotients:
  \begin{listinenv}
  \begin{enumerate}
  \item If $\C$ has pullbacks of $\M$-morphisms
    (i.e.~finite intersections) and $F\colon \C\to \C$ preserves them,
    then $\Coalg_I(F)$ has pullbacks of $\M$-carried homomorphisms. Given a reachable
    subcoalgebra $(D,d,i_D)$ of $(C,c,i_C)$, then it is the least $I$-pointed
    subcoalgebra of $(C,c,i_C)$ (cf.~\cite[Notation~3.18]{amms13}) and is unique
    up to isomorphism.
  \item If $\C$ has pushouts of $\E$-morphisms, then $\Coalg(F)$ has pushouts of
    $\E$-carried homomorphisms. Hence, the simple quotient of a coalgebra
    $(C,c)$ is the greatest quotient of $(C,c)$ and unique up to isomorphism
    (e.g.~\cite[Lemma~2.9]{WissmannDMS19}).
  \end{enumerate}
  \end{listinenv}
\end{instance}

There are instances where a minimization $M$ exists, but where a mediating
morphism in the sense of \autoref{MminLeast} is not unique:

\begin{example}[Tree unravelling]
  \label{treeUnravel}
  Let $\Coalg_I(F)_\reach$ be the category of reachable pointed $F$-coalgebras,
  i.e.~the full subcategory $\Coalg_I(F)_\reach\subseteq \Coalg_I(F)$ such that
  $(C,c,i_C) \in \Coalg_I(F)_\reach$ iff it is reachable. For simplicity,
  restrict to $F\colon \Set\to\Set$ with the $(\Epi,\Mono)$-factorization
  system. Thus, all morphisms in $\Coalg_I(F)_\reach$ are surjective
  (\autoref{minimalInE}). Considering the (trivial) $(\Mor,\Iso)$-factorization
  system on $\Coalg_I(F)_\reach$, a coalgebra $(C,c,i_C)$ is ($\Mor$-)minimal iff
  every coalgebra morphism $h\colon (D,d,i_D)\to (C,c,i_C)$ (with $(D,d,i_D)$
  also reachable) is an isomorphism. If $(C,c,i_C)$ is $\Mor$-minimal, then it
  is a tree: to see this, take $h$ to be its tree unravelling (see e.g.~\autoref{fig:unravel}),
  and by the $\Mor$-minimality, $h$ is an isomorphism, so $(C,c,i_C)$ is already a tree.
  \hackynewpage

  This implies that if the ($\Mor$-)minimization of a coalgebra exists, then it is
  its tree unravelling. For example, for $I=1$ and the bag functor $FX=\Bag X$,
  we have the minimizations as illustrated in \autoref{fig:unravel}. It is easy
  to see that for $FX=\Pow X$ however, no coalgebra (with at least one
  transition) has a $\Mor$-minimization, because one can always duplicate
  successor states.\footnote{%
    It can be conjectured that $\Mor$-minimization of
    reachable $F$-coalgebras exist if $F$ admits precise
    factorizations~\cite[Def.~3.1, 3.4]{WissmannDKH19}.}

  For $FX=\Bag X$, all $\Mor$-minimizations exist, but they are not unique up to
  unique isomorphism.
  Consider the tree unravelling $m\colon M\epito X$ in
  \autoref{fig:unravel:sib}. There is an isomorphism $\phi\colon M\to M$ that
  swaps the two successors of the initial state. Hence, $\phi\neq \id_B$, but
  $m\cdot \phi = m\cdot \id_M$, so $M$ is unique up to isomorphism, but not
  unique up to unique isomorphism.
\end{example}
\begin{figure}
  \begin{subfigure}{.5\textwidth}
        \begin{tikzpicture}[coalgebra,baseline=(q0.base)]
          \node[state] (q0) {$\bullet$};
          \node[state] (q1) at (1.3,0.5) {$\bullet$};
          \node[state] (q2) at (1.3,-0.5) {$\bullet$};
          \draw[transition] ([xshift=-3mm]q0.west) to (q0.west);
          \draw[transition] (q0) to node[above left] {}(q1);
          \draw[transition] (q0) to node[below left] {} (q2);
        \end{tikzpicture}
        \begin{tikzcd}
          {} \arrow[->>]{r} & {}
        \end{tikzcd}
        \begin{tikzpicture}[coalgebra,baseline=(q0.base)]
          \node[state] (q0) {$\bullet$};
          \node[state] (q1) at (1.5,0) {$\bullet$};
          \draw[transition] ([xshift=-3mm]q0.west) to (q0.west);
          \draw[transition,bend left=20] (q0) to node[above] {}(q1);
          \draw[transition,bend right=20] (q0) to node[below] {} (q1);
        \end{tikzpicture}
        \subcaption{Unravelling of siblings}
        \label{fig:unravel:sib}
  \end{subfigure}%
  \hfill
  \begin{subfigure}{.48\textwidth}
        \begin{tikzpicture}[coalgebra,baseline=(q0.base)]
          \node[state] (q0) {$\bullet$};
          \node[state] (q1) at (1,0) {$\bullet$};
          \node[state] (q2) at (2.3,0) {$\cdots$};
          \draw[transition] ([xshift=-3mm]q0.west) to (q0.west);
          \draw[transition] (q0) to node[above] {}(q1);
          \draw[transition] (q1) to node[above] {} (q2);
        \end{tikzpicture}
        \begin{tikzcd}
          {} \arrow[->>]{r} & {}
        \end{tikzcd}
        \begin{tikzpicture}[coalgebra,baseline=(q0.base)]
          \node[state] (q0) {$\bullet$};
          \draw[transition] ([xshift=-3mm]q0.west) to (q0.west);
          \draw[transition,loop right,out=30,in=-30,looseness=10]
              (q0) to node[right] {}(q0);
        \end{tikzpicture}
        \subcaption{Unravelling of a loop}
  \end{subfigure}
  \caption{Tree unravelling for $FX=\Bag X$}
  \label{fig:unravel}
\end{figure}

For proving the existence of an $\M$-minimization, we need to require that $\M$
is a subclass of the monomorphisms in $\K$. Under this assumption, we first
establish the converse of \autoref{uniqueMinimization}:
\begin{lemma}
  \label{leastIsMinimal}
  If $\M\subseteq \Mono$ and if the least $\M$-subobject $M$ of $X$ exists, then
  $M$ is the $\M$-minimization of $X$.
\end{lemma}
\begin{proof}
  Let $m\colon M\monoto X$ be the least $\M$-subobject of $X$, and consider
  $s\colon S\monoto M$ in $\M$. Since $m\cdot s\in \M$, there is some $u\colon
  M\to S$ with $(m\cdot s)\cdot u = m$. Since $m$ is monic, we obtain $s\cdot u
  = \id_M$. Hence, $s$ is a split-epimorphism, and together with
  $s\in\M\subseteq \Mono$, $s$ is an isomorphism.
\end{proof}
\begin{proposition}
  \label{existenceMinimization}
  If $\M\subseteq \Mono$, $\K$ has wide pullbacks of
  $\M$-morphisms, and is $\M$-wellpowered, then
  every object $C$ of $\K$ has an $\M$-minimization.
\end{proposition}
\begin{proof}
  Since $\K$ is $\M$-wellpowered, all the $\M$-carried morphisms $m\colon
  M\monoto C$ form up to isomorphism a set $S$. The wide pullback of all $m\in
  S$ exists in $\K$ by assumption, denote it by $\pr_m\colon P \to M$ for
  $m\colon M\monoto C$. All $m'\in S$ are in $\M$ and so are all $\pr_m$
  by \autoref{EM:pullbackwide}. Hence, $p := m\cdot \pr_m\colon P\monoto C$ for
  an arbitrary $m\in S$ represents an $\M$-subobject, and moreover the least
  $\M$-subobject of $C$, as witnessed by the projections $\pr_m$. By
  \autoref{leastIsMinimal}, $P$ is the minimization of $C$.
\end{proof}
\begin{instance} \label{exExistenceCoalgMinimal}
  This proof directly instantiates to the proofs of the existence of the
  reachable subcoalgebra and simple quotient:
  \begin{listinenv}
    \begin{enumerate}
    \item In the reachability case, let $\M$ be a subclass of the monomorphisms, let
      the base category $\C$ have all $\M$-intersections, and let $F\colon \C\to
      \C$ preserve all intersections. Then the reachable part of a given pointed
      coalgebra $(C,c,i_C)$ is obtained as the intersection of all pointed
      subcoalgebras of $(C,c,i_C)$ \cite{amms13}.

      For $\C=\Set$, and $\M$ being the class of injective maps, all
      intersections exist. The condition that $F\colon \Set\to\Set$ preserves
      all intersections is mild: all finitary functors preserve all
      intersections~(\cite[Proof of Lem.~8.8]{amm18} or \cite[Lem.~2.6.10]{Wissmann2020}) and many non-finitary functors do
      as well, e.g.~the powerset functor. An example of a functor that does not
      preserve all intersections is the filter functor~\cite[Sect.~5.3]{Gumm2001}.

    \item For the existence of simple quotients, let $\E$ be a subclass of
      the epimorphisms and let the base category $\C$ be cocomplete and
      $\E$-cowellpowered. Then every $F$-coalgebra (C,c) has a simple quotient
      given by the wide pushout of all quotient coalgebras
      (\cite[Proposition 3.7]{amms13}, and \cite{Gumm08} for $\C=\Set$).

      Every set has only a set of outgoing surjective maps, so all assumptions
      are met for $\C=\Set$, $\E$ the surjective maps, and every $\Set$-functor $F$.
    \end{enumerate}
  \end{listinenv}
\end{instance}

\begin{remark}
  All observations on simple quotients also apply to pointed coalgebras:
    An $I$-pointed $F$-coalgebra is simple iff it is
      $\E$-carried-minimal in $\K:=\Coalg_I(F)^\dual$. 
      The forgetful functor
      \[
        \Coalg_I(F) \longrightarrow \Coalg(F)
      \]
      preserves and reflects simple coalgebras and simple quotients (note that
      for every pointed coalgebra $(C,c,i_C)$, the slice categories
      $(C,c,i_C)/\Coalg_I(F)$ and $(C,c)/\Coalg(F)$ are isomorphic). For the sake
      of simplicity, we will not state the results explicitly for simple
      coalgebras in $\Coalg_I(F)$.
\end{remark}

\begin{definition}
  We denote by $J\colon \K_{\min}\hookto \K$ the full subcategory formed by the
  $\M$-minimal objects of $\K$.
\end{definition}

In the existence proof of minimal objects
(\autoref{existenceMinimization}) we only required (wide) pullbacks
where all morphisms in the diagram are in $\M$. We obtain additional properties
if we assume the pullback along $\M$-morphisms, i.e.~pullbacks where only one of
the two morphisms is in $\M$:

\begin{proposition}
  \label{MminCorefl}
  Suppose that pullbacks along $\M$-morphisms exist in $\K$ and that every
  object of $\K$ has an $\M$-minimization. Then $J\colon
  \K_{\min}\hookto \K$ is a coreflective subcategory. Its
  right-adjoint $R\colon \K\to \K_{\min}$ ($J\dashv R$) sends an object to its
  $\M$-minimization; in particular, minimization is functorial.
\end{proposition}

\begin{proof}
  The universal property of $J$ follows directly from \autoref{MminLeast}: To
  this end, it suffices to consider $R$ as an object assignment. Given a
  morphism $h\colon M\to X$ where $M$ is $\M$-minimal, we need to show that it
  factorizes uniquely through the $\M$-minimization $s\colon S\monoto X$ of
  $D$, that is $RD := S$. Since the pullback of $h$ along $s$ exists by
  assumption, \autoref{MminLeast} yields us the desired unique factorization
  $u\colon M\to S$ with $s\cdot u = h$.
\end{proof}

\begin{samepage}
\begin{instance} \label{exExistsMinimization}
  For both of our main instances, this adjunction has been observed before:
  \begin{listinenv}
    \begin{enumerate}
    \item If $F\colon \C\to \C$ preserves inverse images (w.r.t.~$\M$), then pullbacks along
      $\M$-carried homomorphisms exist in $\Coalg_I(F)$. Hence, the reachable
      $I$-pointed $F$-coalgebras form a coreflective subcategory of
      $\Coalg_I(F)$, where the coreflector maps a pointed coalgebra to its
      reachable part~\cite[Thm~5.23]{wmkd20reachability}

    \item\label{iExistsSimple} The simple coalgebras form a reflective
      subcategory of $\Coalg(F)$, and the reflector sends a coalgebra to its
      simple quotient, under the assumption that the base category has pushouts
      along $\E$-morphisms. For coalgebras in \Set, the adjunction $J\vdash R$ has been shown by
      Gumm~\cite[Theorem 2.3]{Gumm08}.
    \end{enumerate}
  \end{listinenv}
\end{instance}
\end{samepage}

\begin{corollary}\label{MminQuot}
  If pullbacks along $\M$-morphisms exist in $\K$ and all $\M$-minimizations
  exist, then $\M$-minimal objects are closed under $\E$-quotients.
\end{corollary}
\hackynewpage
\begin{proof}
  Consider an $\E$-morphism $e\colon C\epito D$ where $C$ is $\M$-minimal. Take
  the adjoint transpose $f\colon C\to RD$ with $m\cdot f = e$ where $m\colon
  RD\monoto D$ is the $\M$-minimization of $D$:
  \[
    \begin{tikzcd}
      C
      \arrow[->>]{dr}{e}
      \arrow{d}[swap]{f}
      \\
      RD
      \arrow[>->]{r}{m}
      & D
    \end{tikzcd}
  \]
  Since $RD$ is $\M$-minimal, $f$ is in $\E$ (\autoref{minimalInE}).
  Moreover, $f\in \E$ and $m\cdot f\in \E$ imply $m\in \E$
  (\itemref{rem:EM}{rem:EM:3}), hence $m\in \E\cap\M$ is an isomorphism.
\end{proof}

\begin{example}
  \label{exQuotClosed}
  \begin{listinenv}
    \begin{enumerate}
    \item\label{exQuotClosed:Reach} If $F$ preserves inverse images, then reachable $F$-coalgebras are
      closed under quotients~\cite[Cor.~5.24]{wmkd20reachability}. Note that if
      $F$ does not preserve inverse images, then a quotient of a reachable
      $F$-coalgebra may not be reachable. For example, in (pointed) coalgebras
      for the monoid-valued functor $(\R,+,0)^{(-)}$ there is the coalgebra
      quotient with $h(b_1)=h(b_2) = b$:
      \begin{center}
        \begin{tikzpicture}[coalgebra,baseline=(q0.base)]
          \node[state] (q0) {$a$};
          \node[state] (q1) at (2,-0.5) {$b_1$};
          \node[state] (q2) at (2,0.5) {$b_2$};
          \draw[transition] ([xshift=-3mm]q0.west) to (q0.west);
          \draw[transition] (q0) to node[below left] {3}(q1);
          \draw[transition] (q0) to node[above left] {-3} (q2);

        \end{tikzpicture}
        \quad
        \begin{tikzcd}
          {} \arrow[->>]{r}{h}
          & {}
        \end{tikzcd}
        \quad
        \begin{tikzpicture}[coalgebra,baseline=(q0.base)]
          \node[state] (q0) {$a$};
          \node[state] (q1) at (1,0) {$b$};
          \draw[transition] ([xshift=-3mm]q0.west) to (q0.west);
        \end{tikzpicture}
      \end{center}
      Since transition weights may cancel out each other ($-3 + 3 = 0$), the
      codomain of $h$ is not reachable even though its domain is.
    \item If the base category $\C$ has pushouts along $\E$-morphisms, then
      simple $F$-coalgebras are closed under subcoalgebras.
      For $\C=\Set$, this is obvious: if in a coalgebra $(C,c)$, all states are
      of pairwise different behaviour, then so they are in every subcoalgebra of
      $(C,c)$.
    \end{enumerate}
  \end{listinenv}%
\end{example}

\takeout{}

\section{Interplay of minimality notions}
\label{secInterplay}

The two main aspects of minimization we have seen -- reachability and minimization
for observability -- are closely connected on an abstract level 
and also interact well as we see in the following. In order to minimize a
pointed coalgebra under both aspects, we have two options: first construct the
reachable part and then the simple quotient, or we first form the simple
quotient and then construct its reachable part. Given the existence of pullbacks
of $\M$-morphisms along arbitrary morphisms,
we can show that any order is fine.

In the abstract setting of a category $\K$ with an $(\E,\M)$-factorization
system we are transforming an object $C\in \K$ into an object $C'$ that is
$\M$-minimal in $\K$ and $\E$-minimal in~$\K^\dual$.
\begin{proposition}
  Suppose $\K$ has an $(\E,\M)$-factorization system such that
  all $\M$-mi\-nim\-iza\-tions in $\K$ and all $\E$-minimizations in $\K^\dual$ exist.
  If $\K$ has pullbacks along $\M$-morphisms and pushouts along $\E$-morphisms,
  then for every $C$ in $\K$ the following two constructions yield the same object:
  \begin{enumerate}
  \item The $\M$-minimization of $C$ in $\K$ followed by its $\E$-minimization
    in $\K^\dual$.
  \item The $\E$-minimization of $C$ in $\K^\dual$ followed by its $\M$-minimization
    in $\K$.
  \end{enumerate}
\end{proposition}
\begin{proof}
  In the first approach,
  denote the $\M$-minimization of $C$ by $m\colon R\monoto C$ and its
  $\E$-minimization by $s\colon R \epito V$. In the other approach, denote the
  $\E$-minimization of $C$ by $e\colon C\epito Q$ and its $\M$-minimization by
  $t\colon W\monoto Q$:
  \[
    \begin{tikzcd}
      R
      \arrow[>->]{rr}{m}
      \arrow[->>]{d}[swap]{s}
      & & C
      \arrow[->>]{d}{e}
      \\
      V
      & W
      \arrow[>->]{r}{t}
      & Q
    \end{tikzcd}
  \]
  We need to prove that $V$ and $W$ are isomorphic, making the above
  (then-closed) square commute. The $\M$-minimal objects form a coreflective
  subcategory (\autoref{MminCorefl}), so $e\cdot m$, whose domain is
  $\M$-minimal, factorizes through the $\M$-minimization of the codomain of
  $e\cdot m$, i.e.~we have $h\colon R\to W$ with $t\cdot h = e\cdot m$. Since
  $Q$ is $\E$-minimal, its $\M$-subobject $W$ is also $\E$-minimal in $\K^\dual$ (\autoref{MminQuot}). The
  $\E$-minimal objects form a reflective subcategory (\autoref{MminCorefl}).
  Applying the reflection to $h\colon R\to W$, we obtain $\phi\colon V\to W$ with $h =
  \phi\cdot s$. Since $V$ is $\E$-minimal (in $\K^\dual$), $\phi$ is in $\M$,
  and since $W$ is $\M$-minimal, $\phi$ is in $\E$, and thus $\phi$ is an isomorphism.
\end{proof}
\noindent In the concrete case of $F$-coalgebras, a coalgebra that is both simple and reachable is called a \emph{well-pointed}
coalgebra (see \cite[Section 3.2]{amms13}).
\begin{instance}
  If $F\colon \C\to \C$ fulfils all assumptions from the previous
  \autoref{exExistsMinimization} (and in particular preserves inverse images),
  then the construction of the simple quotient and the reachability construction
  for $F$-coalgebras can be performed in any order, yielding the same well-pointed
  coalgebra. 
\end{instance}

If $F$ does not preserve inverse images, then in the construction of the simple
quotient, transitions may cancel out each other and this may affect the
reachability of states. We have seen an example for this in
\itemref{exQuotClosed}{exQuotClosed:Reach} where performing reachability first
and observability second leads to a simple coalgebra in which states are unreachable, i.e.~the result is not well-pointed.
Hence, in
contrast to the well-known automata minimization procedure, the minimization of
a coalgebra in general has to be performed by first computing its simple quotient and secondly
computing the reachable part in the simple quotient.

If $F$ preserves inverse images, such as the functor for automata, any order is
fine. In sets, the reachability computation is a simple breadth-first
search~\cite{wmkd20reachability}, and hence runs in linear time. On the other
hand, existing algorithms for computing the simple quotient for many $\Set$-functors run
in at least $n\cdot \log n$ time where $n$ is the size of the
coalgebra~\cite{GrooteEA21,WissmannDMS19}. Hence, the reachability analysis
should be done first whenever possible.

\section{Conclusions}
We have seen a common ground for minimality notions in a category with various
instances in a coalgebraic setting. The abstract
results about the uniqueness and the existence of the minimization instantiate
to the standard results for reachability and observability of coalgebras.
Most of the general results even hold if the $(\E,\M)$-factorization system is
not proper. The tree unravelling of an automaton is an instance of
minimization for a non-proper factorization system.

It remains for future work to relate the efficient algorithmic approaches to the
minimization tasks: reachability is computed by breadth-first
search~\cite{wmkd20reachability,BarloccoKR19} and observability is computed by
partition refinement algorithms~\cite{KonigK14,WissmannDMS19,DorschMSW17}.
Even though their run-time complexity differs -- reachability is usually linear,
whereas partition refinement algorithms are quasilinear or slower -- they have
striking similarities.
All these algorithms compute a chain of subobjects resp.~quotients on the carrier of
the input coalgebra and terminate at the first element of the chain admitting a
coalgebra structure compatible with the input coalgebra. It is thus likely that
this relation can be made formal. A similar connection between the reachability of
algebras and partition refinement on coalgebras is already known~\cite{Rot16}.

\bibliography{refs}

\begin{thebibliography}{10}

\bibitem{aczelmendler:89}
Peter Aczel and Nax Mendler.
\newblock A final coalgebra theorem.
\newblock In {\em Proc.~Category Theory and Computer Science (CTCS)}, volume
  389 of {\em Lecture Notes Comput.~Sci.}, pages 357--365. Springer, 1989.

\bibitem{joyofcats}
Ji\v{r}\'i\ Ad{\'{a}}mek, Horst Herrlich, and George~E.\ Strecker.
\newblock {\em Abstract and Concrete Categories: The Joy of Cats}.
\newblock Dover Publications, 2nd edition, 2009.

\bibitem{amm18}
Ji\v{r}\'i Ad\'amek, Stefan Milius, and Lawrence~S.\ Moss.
\newblock Fixed points of functors.
\newblock {\em Journal of Logical and Algebraic Methods in Programming},
  95:41--81, 2018.

\bibitem{amms13}
Ji\v{r}\'i Ad\'amek, Stefan Milius, Lawrence~S. Moss, and Lurdes Sousa.
\newblock Well-pointed coalgebras.
\newblock {\em Logical Methods in Computer Science}, 9(3:2):51 pp., 2013.

\bibitem{awodey2010category}
Steve Awodey.
\newblock {\em Category Theory}.
\newblock Oxford Logic Guides. OUP Oxford, 2010.

\bibitem{BalanK11}
Adriana Balan and Alexander Kurz.
\newblock Finitary functors: From set to preord and poset.
\newblock In Andrea Corradini, Bartek Klin, and Corina C{\^{\i}}rstea, editors,
  {\em Algebra and Coalgebra in Computer Science - 4th International
  Conference, {CALCO} 2011, Winchester, UK, August 30 - September 2, 2011.
  Proceedings}, volume 6859 of {\em Lecture Notes in Computer Science}, pages
  85--99. Springer, 2011.
\newblock \href {https://doi.org/10.1007/978-3-642-22944-2_7}
  {\path{doi:10.1007/978-3-642-22944-2_7}}.

\bibitem{BarloccoKR19}
Simone Barlocco, Clemens Kupke, and Jurriaan Rot.
\newblock Coalgebra learning via duality.
\newblock In Mikolaj Bojanczyk and Alex Simpson, editors, {\em Foundations of
  Software Science and Computation Structures - 22nd International Conference,
  {FOSSACS} 2019, Held as Part of the European Joint Conferences on Theory and
  Practice of Software, {ETAPS} 2019, Prague, Czech Republic, April 6-11, 2019,
  Proceedings}, volume 11425 of {\em Lecture Notes in Computer Science}, pages
  62--79. Springer, 2019.
\newblock \href {https://doi.org/10.1007/978-3-030-17127-8_4}
  {\path{doi:10.1007/978-3-030-17127-8_4}}.

\bibitem{BartelsSV04}
Falk Bartels, Ana Sokolova, and Erik~P. de~Vink.
\newblock A hierarchy of probabilistic system types.
\newblock {\em Theor. Comput. Sci.}, 327(1-2):3--22, 2004.
\newblock \href {https://doi.org/10.1016/j.tcs.2004.07.019}
  {\path{doi:10.1016/j.tcs.2004.07.019}}.

\bibitem{BezhanishviliKP12}
Nick Bezhanishvili, Clemens Kupke, and Prakash Panangaden.
\newblock Minimization via duality.
\newblock In C.{-}H.~Luke Ong and Ruy J. G.~B. de~Queiroz, editors, {\em Logic,
  Language, Information and Computation - 19th International Workshop, WoLLIC
  2012, Buenos Aires, Argentina, September 3-6, 2012. Proceedings}, volume 7456
  of {\em Lecture Notes in Computer Science}, pages 191--205. Springer, 2012.
\newblock \href {https://doi.org/10.1007/978-3-642-32621-9_14}
  {\path{doi:10.1007/978-3-642-32621-9_14}}.

\bibitem{BidoitHK01}
Michel Bidoit, Rolf Hennicker, and Alexander Kurz.
\newblock On the duality between observability and reachability.
\newblock In Furio Honsell and Marino Miculan, editors, {\em Foundations of
  Software Science and Computation Structures, 4th International Conference
  ({FOSSACS} 2001), Held as Part of {ETAPS} 2001 Genova, Italy, April 2-6,
  2001, Proceedings}, volume 2030 of {\em Lecture Notes in Computer Science},
  pages 72--87. Springer, 2001.
\newblock \href {https://doi.org/10.1007/3-540-45315-6_5}
  {\path{doi:10.1007/3-540-45315-6_5}}.

\bibitem{BonchiBHPRS14}
Filippo Bonchi, Marcello~M. Bonsangue, Helle~Hvid Hansen, Prakash Panangaden,
  Jan J. M.~M. Rutten, and Alexandra Silva.
\newblock Algebra-coalgebra duality in brzozowski's minimization algorithm.
\newblock {\em {ACM} Trans. Comput. Log.}, 15(1):3:1--3:29, 2014.
\newblock \href {https://doi.org/10.1145/2490818} {\path{doi:10.1145/2490818}}.

\bibitem{DerisaviHS03}
Salem Derisavi, Holger Hermanns, and William~H. Sanders.
\newblock Optimal state-space lumping in markov chains.
\newblock {\em Inf. Process. Lett.}, 87(6):309--315, 2003.
\newblock \href {https://doi.org/10.1016/S0020-0190(03)00343-0}
  {\path{doi:10.1016/S0020-0190(03)00343-0}}.

\bibitem{DorschMSW17}
Ulrich Dorsch, Stefan Milius, Lutz Schr{\"{o}}der, and Thorsten Wi{\ss}mann.
\newblock Efficient coalgebraic partition refinement.
\newblock In Roland Meyer and Uwe Nestmann, editors, {\em 28th International
  Conference on Concurrency Theory, {CONCUR} 2017, September 5-8, 2017, Berlin,
  Germany}, volume~85 of {\em LIPIcs}, pages 32:1--32:16. Schloss Dagstuhl -
  Leibniz-Zentrum f{\"{u}}r Informatik, 2017.
\newblock \href {https://doi.org/10.4230/LIPIcs.CONCUR.2017.32}
  {\path{doi:10.4230/LIPIcs.CONCUR.2017.32}}.

\bibitem{DovierPP04}
Agostino Dovier, Carla Piazza, and Alberto Policriti.
\newblock An efficient algorithm for computing bisimulation equivalence.
\newblock {\em Theor. Comput. Sci.}, 311(1-3):221--256, 2004.
\newblock \href {https://doi.org/10.1016/S0304-3975(03)00361-X}
  {\path{doi:10.1016/S0304-3975(03)00361-X}}.

\bibitem{GrooteEA21}
Jan~Friso Groote, Jan Martens, and Erik de~Vink.
\newblock {Bisimulation by Partitioning Is $\Omega((m+n)\log n)$}.
\newblock In Serge Haddad and Daniele Varacca, editors, {\em 32nd International
  Conference on Concurrency Theory (CONCUR 2021)}, volume 203 of {\em Leibniz
  International Proceedings in Informatics (LIPIcs)}, pages 31:1--31:16,
  Dagstuhl, Germany, 2021. Schloss Dagstuhl -- Leibniz-Zentrum f{\"u}r
  Informatik.
\newblock \href {https://doi.org/10.4230/LIPIcs.CONCUR.2021.31}
  {\path{doi:10.4230/LIPIcs.CONCUR.2021.31}}.

\bibitem{Gumm2001}
H.~Peter Gumm.
\newblock Functors for coalgebras.
\newblock {\em Algebra Universalis}, 45(2):135--147, April 2001.
\newblock \href {https://doi.org/10.1007/s00012-001-8156-x}
  {\path{doi:10.1007/s00012-001-8156-x}}.

\bibitem{Gumm08}
H.~Peter Gumm.
\newblock On minimal coalgebras.
\newblock {\em Applied Categorical Structures}, 16(3):313--332, June 2008.
\newblock \href {https://doi.org/10.1007/s10485-007-9116-1}
  {\path{doi:10.1007/s10485-007-9116-1}}.

\bibitem{GummS01}
H.~Peter Gumm and Tobias Schr{\"{o}}der.
\newblock Monoid-labeled transition systems.
\newblock In {\em Coalgebraic Methods in Computer Science, {CMCS} 2001}, volume
  44(1) of {\em ENTCS}, pages 185--204. Elsevier, 2001.

\bibitem{GS05}
H.~Peter Gumm and Tobias Schr{\"o}der.
\newblock Types and coalgebraic structure.
\newblock {\em algebra universalis}, 53(2):229--252, 2005.
\newblock \href {https://doi.org/10.1007/s00012-005-1888-2}
  {\path{doi:10.1007/s00012-005-1888-2}}.

\bibitem{HasuoJS06}
Ichiro Hasuo, Bart Jacobs, and Ana Sokolova.
\newblock Generic trace theory.
\newblock In Neil Ghani and John Power, editors, {\em Proceedings of the Eighth
  Workshop on Coalgebraic Methods in Computer Science, {CMCS} 2006, Vienna,
  Austria, March 25-27, 2006}, volume 164 of {\em Electronic Notes in
  Theoretical Computer Science}, pages 47--65. Elsevier, 2006.
\newblock \href {https://doi.org/10.1016/j.entcs.2006.06.004}
  {\path{doi:10.1016/j.entcs.2006.06.004}}.

\bibitem{Hopcroft71}
John Hopcroft.
\newblock An $n \log n$ algorithm for minimizing states in a finite automaton.
\newblock In {\em Theory of Machines and Computations}, pages 189--196.
  Academic Press, 1971.

\bibitem{Gumm03}
Thomas Ihringer.
\newblock {\em Algemeine Algebra. Mit einem Anhang \"{u}ber Universelle
  Coalgebra von H.~P.~Gumm}, volume~10 of {\em Berliner Studienreihe zur
  Mathematik}.
\newblock Heldermann Verlag, 2003.

\bibitem{KlinS13}
Bartek Klin and Vladimiro Sassone.
\newblock Structural operational semantics for stochastic and weighted
  transition systems.
\newblock {\em Inf.~Comput.}, 227:58--83, 2013.

\bibitem{KonigK14}
Barbara K{\"{o}}nig and Sebastian K{\"{u}}pper.
\newblock Generic partition refinement algorithms for coalgebras and an
  instantiation to weighted automata.
\newblock In Josep D{\'{\i}}az, Ivan Lanese, and Davide Sangiorgi, editors,
  {\em Theoretical Computer Science - 8th {IFIP} {TC} 1/WG 2.2 International
  Conference, {TCS} 2014, Rome, Italy, September 1-3, 2014. Proceedings},
  volume 8705 of {\em Lecture Notes in Computer Science}, pages 311--325.
  Springer, 2014.
\newblock \href {https://doi.org/10.1007/978-3-662-44602-7_24}
  {\path{doi:10.1007/978-3-662-44602-7_24}}.

\bibitem{kurzPhd}
Alexander Kurz.
\newblock {\em Logics for Coalgebras and Applications to Computer Science}.
\newblock PhD thesis, Ludwig-Maximilians-Universität München, 7 2000.
\newblock \url{https://www.cs.le.ac.uk/people/akurz/LMU/Diss/all-s.ps.gz}.

\bibitem{KurzPSV13}
Alexander Kurz, Daniela Petrisan, Paula Severi, and Fer{-}Jan de~Vries.
\newblock Nominal coalgebraic data types with applications to lambda calculus.
\newblock {\em Log. Methods Comput. Sci.}, 9(4), 2013.
\newblock \href {https://doi.org/10.2168/LMCS-9(4:20)2013}
  {\path{doi:10.2168/LMCS-9(4:20)2013}}.

\bibitem{lffiac}
Stefan Milius, Dirk Pattinson, and Thorsten Wi{\ss}mann.
\newblock A new foundation for finitary corecursion and iterative algebras.
\newblock {\em Information and Computation}, page 104456, 09 2019.
\newblock \href {https://doi.org/10.1016/j.ic.2019.104456}
  {\path{doi:10.1016/j.ic.2019.104456}}.

\bibitem{MiliusSW16}
Stefan Milius, Lutz Schr{\"{o}}der, and Thorsten Wi{\ss}mann.
\newblock Regular behaviours with names - on rational fixpoints of endofunctors
  on nominal sets.
\newblock {\em Appl. Categorical Struct.}, 24(5):663--701, 2016.
\newblock \href {https://doi.org/10.1007/s10485-016-9457-8}
  {\path{doi:10.1007/s10485-016-9457-8}}.

\bibitem{PaigeT87}
Robert Paige and Robert~Endre Tarjan.
\newblock Three partition refinement algorithms.
\newblock {\em {SIAM} J. Comput.}, 16(6):973--989, 1987.
\newblock \href {https://doi.org/10.1137/0216062} {\path{doi:10.1137/0216062}}.

\bibitem{Rot16}
Jurriaan Rot.
\newblock Coalgebraic minimization of automata by initiality and finality.
\newblock In Lars Birkedal, editor, {\em The Thirty-second Conference on the
  Mathematical Foundations of Programming Semantics, {MFPS} 2016, Carnegie
  Mellon University, Pittsburgh, PA, USA, May 23-26, 2016}, volume 325 of {\em
  Electronic Notes in Theoretical Computer Science}, pages 253--276. Elsevier,
  2016.
\newblock \href {https://doi.org/10.1016/j.entcs.2016.09.042}
  {\path{doi:10.1016/j.entcs.2016.09.042}}.

\bibitem{Rutten00}
Jan J. M.~M. Rutten.
\newblock Universal coalgebra: a theory of systems.
\newblock {\em Theor. Comput. Sci.}, 249(1):3--80, 2000.
\newblock \href {https://doi.org/10.1016/S0304-3975(00)00056-6}
  {\path{doi:10.1016/S0304-3975(00)00056-6}}.

\bibitem{trnkova71}
V\v{e}ra Trnkov\'a.
\newblock On a descriptive classification of set functors {I}.
\newblock {\em Commentationes Mathematicae Universitatis Carolinae},
  12(1):143--174, 1971.

\bibitem{TuriP97}
Daniele Turi and Gordon~D. Plotkin.
\newblock Towards a mathematical operational semantics.
\newblock In {\em Proceedings, 12th Annual {IEEE} Symposium on Logic in
  Computer Science, Warsaw, Poland, June 29 - July 2, 1997}, pages 280--291.
  {IEEE} Computer Society, 1997.
\newblock \href {https://doi.org/10.1109/LICS.1997.614955}
  {\path{doi:10.1109/LICS.1997.614955}}.

\bibitem{Valmari09}
Antti Valmari.
\newblock Bisimilarity minimization in o(m logn) time.
\newblock In Giuliana Franceschinis and Karsten Wolf, editors, {\em
  Applications and Theory of Petri Nets, 30th International Conference, {PETRI}
  {NETS} 2009, Paris, France, June 22-26, 2009. Proceedings}, volume 5606 of
  {\em Lecture Notes in Computer Science}, pages 123--142. Springer, 2009.
\newblock \href {https://doi.org/10.1007/978-3-642-02424-5_9}
  {\path{doi:10.1007/978-3-642-02424-5_9}}.

\bibitem{ValmariF10}
Antti Valmari and Giuliana Franceschinis.
\newblock Simple \emph{O}(\emph{m} log\emph{n}) time markov chain lumping.
\newblock In Javier Esparza and Rupak Majumdar, editors, {\em Tools and
  Algorithms for the Construction and Analysis of Systems, 16th International
  Conference, {TACAS} 2010, Held as Part of the Joint European Conferences on
  Theory and Practice of Software, {ETAPS} 2010, Paphos, Cyprus, March 20-28,
  2010. Proceedings}, volume 6015 of {\em Lecture Notes in Computer Science},
  pages 38--52. Springer, 2010.
\newblock \href {https://doi.org/10.1007/978-3-642-12002-2_4}
  {\path{doi:10.1007/978-3-642-12002-2_4}}.

\bibitem{WissmannDMS19}
Thorsten Wi{\ss}mann, Ulrich Dorsch, Stefan Milius, and Lutz Schr{\"{o}}der.
\newblock Efficient and modular coalgebraic partition refinement.
\newblock {\em Log. Methods Comput. Sci.}, 16(1), 2020.
\newblock \href {https://doi.org/10.23638/LMCS-16(1:8)2020}
  {\path{doi:10.23638/LMCS-16(1:8)2020}}.

\bibitem{WissmannDKH19}
Thorsten Wi{\ss}mann, J{\'{e}}r{\'{e}}my Dubut, Shin{-}ya Katsumata, and Ichiro
  Hasuo.
\newblock Path category for free - open morphisms from coalgebras with
  non-deterministic branching.
\newblock In Mikolaj Bojanczyk and Alex Simpson, editors, {\em Foundations of
  Software Science and Computation Structures - 22nd International Conference
  ({FOSSACS} 2019), Held as Part of {ETAPS} 2019, Prague, Czech Republic, April
  6-11, 2019, Proceedings}, volume 11425 of {\em Lecture Notes in Computer
  Science}, pages 523--540. Springer, 2019.
\newblock \href {https://doi.org/10.1007/978-3-030-17127-8_30}
  {\path{doi:10.1007/978-3-030-17127-8_30}}.

\bibitem{Wissmann2020}
Thorsten Wißmann.
\newblock {\em Coalgebraic Semantics and Minimization in Sets and Beyond}.
\newblock Phd thesis, Friedrich-Alexander-Universit{\"a}t Erlangen-N{\"u}rnberg
  (FAU), 2020.
\newblock URL:
  \url{https://opus4.kobv.de/opus4-fau/frontdoor/index/index/docId/14222}.

\bibitem{wmkd20reachability}
Thorsten Wißmann, Stefan Milius, Shin-ya Katsumata, and Jérémy Dubut.
\newblock A coalgebraic view on reachability.
\newblock {\em Commentationes Mathematicae Universitatis Carolinae},
  60:4:605--638, 12 2019.
\newblock \href {https://doi.org/10.14712/1213-7243.2019.026}
  {\path{doi:10.14712/1213-7243.2019.026}}.

\end{thebibliography}

\ifthenelse{\boolean{withappendix}}{
\appendix
\section*{Appendix: Omitted Proofs}
\addcontentsline{toc}{section}{Appendix}%

\begin{proofappendix}{cor:coalg-colims}
  Let $c_i\colon UDi\to FUDi$ be the coalgebra structure of $Di\in \Coalg(F)$
  for every $i\in \D$.
  For the colimit of $UD\colon \D\to \C$
  \[
    \begin{tikzcd}
      UDi
      \arrow{r}{\inj_i}
      & \colim (UD)
    \end{tikzcd}
    \qquad\text{for every }i\in \D
  \]
  apply $F$ and precompose with $c_i$, yielding
  \[
    \begin{tikzcd}
      UDi
      \arrow{r}{c_i}
      &
      FUDi
      \arrow{r}{F\inj_i}
      & F\colim (UD)
    \end{tikzcd}
    \qquad\text{for every }i\in \D.
  \]
  This is a cocone for the diagram $D$ because for all $h\colon i\to j$ in $\D$
  the outside of the following diagram commutes:
  \[
    \begin{tikzcd}
      UDi
      \arrow{r}{c_i}
      \descto{dr}{$Dh$ coalgebra \\ morphism}
      \arrow{d}[swap]{UDh}
      &[8mm]
      FUDi
      \arrow{rd}{F\inj_i}
      \arrow{d}{FUDh}
      \\
      UDj
      \arrow{r}[swap]{c_j}
      &
      FUDj
      \arrow{r}[swap]{F\inj_j}
      & F\colim (UD)
    \end{tikzcd}
  \]
  Thus we obtain a coalgebra structure $u\colon \colim(UD)\to F\colim(UD)$.
  Since $u$ is a cocone-morphism, every $\inj_i$ is an $F$-coalgebra morphism.

  In order to show that $(\colim(UD), u)$ is the colimit of $D\colon \D\to
  \Coalg(F)$, consider another cocone $(m_i\colon Di\to (E,e))_{i\in \D}$.
  \[
    \begin{tikzcd}
      \colim (UD)
      \arrow{d}[swap]{u}
      \arrow[dashed]{r}{w}
      & E \arrow{d}{e} \\
      F\colim(UD)
      \arrow{r}{Fw}
      & FE
    \end{tikzcd}
  \]
  In $\C$, we obtain a cocone morphism $w\colon \colim (UD)\to E$. With a
  similar verification as before,  $(e\cdot m_i\colon UDi\to FE)_{i\in \D}$ is a
  cocone for $D$, and thus both $d\cdot w$ and $Fw\cdot u\colon \colim (UD)\to
  FE$ are cocone morphisms (for $UD$). Since $\colim (UD)$ is the colimit, this
  implies that $d\cdot w = Fw\cdot u$, i.e.~$w\colon (colim (UD), u)\to (E,e)$ is
  a coalgebra morphism. Since $U\colon \Coalg(F)\to \C$ is faithful, $w$ is the
  unique cocone morphism, and so $(\colim UD,u)$ is indeed the colimit of~$UD$.
  \qed
\end{proofappendix}

\begin{proofappendix}{EM:pullbackwide}
    Consider the $(\E,\M)$-factorization of $\pr_j$ into $e\colon
    P\epito C$ and $m\colon C\monoto A_j$ with $\pr_j = m\cdot e$. On the image,
    we define a cone structure $(c_i\colon C\to A_i)_{i\in I}$ by $c_j = m$ and
    for every $i\in I\setminus\{j\}$ by the diagonal fill-in:
    \[
      \begin{tikzcd}
        P
        \arrow[->>]{r}{e}
        \arrow{d}[swap]{\pr_i}
        & C
        \arrow[>->]{r}{c_j}
        \arrow[dashed]{dl}{c_i}
        & A_j
        \arrow{d}{f_j}
        \\
        A_i
        \arrow[>->]{rr}{f_i}
        & & B
      \end{tikzcd}
      \quad\text{for all }i\in I\setminus\{j\}.
    \]
    The diagonal $c_i$ is induced, because $f_i\in \M$ for all $i\in I\setminus\{j\}$.
    The family $(c_i)_{i\in I}$ forms a cone for the wide pullback, because for
    all $i,i'\in I$ we have $f_i\cdot c_i = f_j\cdot c_j = f_{i'}\cdot c_{i'}$.
    This makes $e$ a cone morphism, because $c_i\cdot e = \pr_i$ for
    all $i\in I$.
    Moreover, the limiting cone $P$ induces a cone morphism $s\colon C\to P$ and
    we have $s\cdot e = \id_P$. Consider the commutative diagrams:
    \[
      \begin{tikzcd}
        P
        \arrow[->>]{rr}{e}
        \arrow{dd}[swap]{e}
        \arrow{dr}[description,shape=circle,inner sep=1pt]{\id_P}
        & & C
        \arrow{dl}[description,shape=circle,inner sep=2pt]{s}
        \arrow{dd}{c_j}
        \\
        & P
        \arrow{dr}[description,shape=circle,inner sep=2pt]{\pr_j}
        \arrow{dl}[description,shape=circle,inner sep=2pt]{e}
        \descto{r}{\((*)\)}
        \descto{d}{\((*)\)}
        & {}
        \\
        C
        \arrow[>->]{rr}[swap]{c_j}
        & {} & A_j
      \end{tikzcd}
      \quad\text{ and }\quad
      \begin{tikzcd}
        P
        \arrow[->>]{rr}{e}
        \arrow{dd}[swap]{e}
        & & C
        \arrow{dd}{c_j}
        \arrow{ddll}[description,shape=circle,inner sep=1pt]{\id_C}
        \\
        & \phantom{P}
        \\
        C
        \arrow[>->]{rr}{c_j}
        & {} & A_j.
      \end{tikzcd}
    \]
    The parts marked by $(*)$ commute because $e$ and $s$ are cone morphisms.
    Since the diagonal fill-in in \itemref{D:factSystem}{diagonalization} is unique, we have $e\cdot s = \id_C$. Thus,
    $e$ is an isomorphism, and $\pr_j = c_j\cdot e$ is in $\M$, as desired.
    \qed
\end{proofappendix}

\begin{proofappendix}{coalgMono}
  It is clear that every $\Mono$-carried homomorphism is monic in $\Coalg(F)$.
  Conversely, let $m\colon (C,c)\to (D,d)$ be a monomorphism in $\Coalg(F)$. Let
  $\pr_1,\pr_2\colon K\to C$ be a weak kernel pair of $m$. Since $F$ preserves
  weak kernel pairs, $F\pr_1,F\pr_2\colon FK\to FC$ is a weak kernel pair of
  $Fm\colon FC\to FD$. This induces some cone morphism $k\colon K\to FK$ making
  $\pr_1$ and $\pr_2$ coalgebra morphisms $(K,k)\to (C,c)$:
  \[
    \begin{tikzcd}
      K
      \arrow[shift left=1]{r}{\pr_1}
      \arrow[shift right=1]{r}[swap]{\pr_2}
      \arrow[dashed]{d}{k}
      & C
      \arrow{d}{c}
      \arrow{r}{m}
      & D
      \arrow{d}{d}
      \\
      FK
      \arrow[shift left=1]{r}{F\pr_1}
      \arrow[shift right=1]{r}[swap]{F\pr_2}
      & FC
      \arrow{r}{Fm}
      & FD
    \end{tikzcd}
  \]
  Since $m$ is monic in $\Coalg(F)$, this implies that $\pr_1=\pr_2$. For the verification that $m$ is a monomorphism in $\C$, consider $f,g\colon X\to C$ with $m\cdot f=m\cdot g$. Since $\pr_1,\pr_2$ is a weak kernel pair, it induces some cone morphism $v\colon X\to K$, fulfilling $f= \pr_1\cdot v$ and $g=\pr_2\cdot v$. Since, $\pr_1=\pr_2$, we find $f=g$
  as desired.
  \qed
\end{proofappendix}

\begin{proofappendix}{monadfactor}
  Denote the unit and multiplication of the monad $T$ by $\eta\colon \Id\to T$
  and $\mu\colon TT\to T$, respectively.
  Consider an $T$-algebra homomorphism $f\colon (A,a)\to (B,b)$ for Eilenberg-Moore
  algebras $(A,a)$ and $(B,b)$ and denote its image factorization in $\Alg(T)$
  by $(I,i)$, with homomorphisms $e\colon (A,a)\epito (I,i)$ and $m\colon
  (I,i)\monoto (B,b)$. We verify that $i\colon TI\to I$ is an Eilenberg-Moore algebra.
  \begin{itemize}
  \item First, we verify $i\cdot \eta_I = \id_I$ by showing that both $i\cdot
    \eta_I$ and $\id_I$ are both diagonals of the following square:

    \[
      \begin{tikzcd}
        A
        \arrow[->>]{rr}{e}
        \arrow[->]{dd}[swap]{e}
        & & I
        \arrow[->]{dd}{m}
        \arrow[->]{ddll}{\id_I}
        \\
        \\
        I
        \arrow[>->]{rr}{m}
        && B
      \end{tikzcd}
      \qquad
      \begin{tikzcd}
        A
        \arrow[->>]{rrr}{e}
        \arrow{dr}{\eta_A}
        \arrow{dd}[swap]{\id_A}
        & {} \descto{dr}{(N)}
        & & I
        \arrow{dl}{\eta_I}
        \arrow{d}{m}
        \\
        {} \descto{r}{(A)}
        & TA
        \arrow{r}{Te}
        \arrow{dl}{a}
        & TI
        \descto{r}{(N)}
        \arrow{d}{Tm}
        \arrow{dldl}[pos=0.60]{i}
        & B
        \arrow{dl}{\eta_B}
        \arrow{dd}{\id_B}
        \\
        A
        \arrow{d}[swap]{e}
        &
        \descto{u}{(H)}
        \descto{r}{(H)}
        & TB
        \arrow{dr}{b}
        & \descto{l}{(A)}
        \\
        I
        \arrow[>->]{rrr}{m}
        & & & B
      \end{tikzcd}
    \]
    The left-hand square commutes trivially, and the right-hand square commutes
    because $\eta$ is natural (N), because $e$ and $m$ are $T$-algebra
    homomorphisms (H), and because $(A,a)$ and $(B,b)$ are Eilenberg-Moore
    algebras (A).
    By the uniqueness of the diagonal lifting property, we obtain $i\cdot \eta_I
    = \id_I$.

  \item Next, we verify $i\cdot Ti = i\cdot \mu_i$ by showing that both are
    diagonals in same square. To this end, let $(s_A,s_I,s_B) \in \{ (Ta,Ti,Tb),
    (\mu_A,\mu_I,\mu_B)\}$, i.e.~we obtain one diagonal for $(Ta,Ti,Tb)$ and one
    for $(\mu_A,\mu_I,\mu_B)$:
    \[
      \begin{tikzcd}
        TTA
        \arrow[->>]{rrr}{TTe}
        \arrow{d}[swap]{\mu_A}
        \arrow{dr}{s_A}
        & {} \descto{dr}{(N)} & & TTI
        \arrow{dl}[swap]{s_I}
        \arrow{d}{TTm}
        \\
        TA
        \descto{r}{(A)}
        \arrow{d}[swap]{a}
        & TA
        \arrow{r}{Te}
        \arrow{ld}[swap]{a}
        & TI
        \arrow{dldl}[pos=0.60]{i}
        \arrow{d}[description]{Tm}
        \descto{r}{(N)}
        & TTB
        \arrow{dl}[description]{s_B}
        \arrow{d}{\mu_B}
        \\
        A
        \arrow{d}[swap]{e}
        & {} \descto{u}{(H)}
        \descto{r}{(H)}
        & TB
        \arrow{rd}[swap]{b}
        \descto{r}{(A)}
        & TB
        \arrow{d}{b}
        \\
        I
        \arrow[>->]{rrr}{m}
        & & & B
      \end{tikzcd}
    \]
    In both cases, we have that $e$ and $m$ are $T$-algebra homomorphisms $(H)$.
    For $(s_A,s_I,s_B) := (Ta,Ti,Tb)$, we have that $(A,a)$ and
    $(B,b)$ are Eilenberg-Moore algebras (A) and that $e$ and $m$ are
    homomorphisms (N). For $(s_A,s_I,s_B) := (\mu_A,\mu_I,\mu_B)$, we have that
    the parts (A) commute trivially, and the parts (N) commute because $\mu$ is natural.
    Since $T$ preserves $\E$, we have $TTe\in \E$, and thus
    by the uniqueness of the diagonal, we obtain $i\cdot Ti = i\cdot \mu_I$.
  \end{itemize}
  Thus, $(I,i)$ fulfils the axioms of an Eilenberg-Moore algebra. The remaining
  properties of the factorization system hold because the Eilenberg-Moore
  category is a full subcategory of $\Alg(T)$.
  \qed
\end{proofappendix}

\begin{proofappendix}{minimalSubterminal}
  We verify the postulated equivalence using that $C$ is $\M$-minimal iff every
  $h\colon B\to C$ is an epimorphism (\autoref{minimalInE}, $\E=\Epi$).
  \begin{itemize}
  \item For \textqt{if}, we verify that every $h\colon B\to C$ is an
    epimorphism: for $u,v\colon C\to D$ with $u\cdot h=v\cdot h$, we directly
    obtain $u=v$ by assumption. Thus, $h$ is an epimorphism.
  \item For \textqt{only if}, consider $u,v\colon  C\to D$ and take a weak
    equalizer $e\colon E\to C$; hence, $u\cdot e= v\cdot e$. Since $e$ is an
    epimorphism (by minimality), we obtain $u=v$.
  \qed
  \end{itemize}
\end{proofappendix}

}{}

\end{document}